\documentclass[11pt,a4paper]{article}
\usepackage[margin=1in]{geometry}
\usepackage{amsmath,amssymb,amsthm,mathtools}
\usepackage[expansion=false]{microtype}
\usepackage{booktabs}
\usepackage{needspace}
\usepackage{enumitem}
\usepackage{xcolor}
\usepackage{graphicx}
\usepackage{algorithm}
\usepackage[noend]{algpseudocode}
\definecolor{linkcolor}{rgb}{0.05,0.20,0.50}
\usepackage[colorlinks=true,linkcolor=linkcolor,citecolor=linkcolor,urlcolor=linkcolor]{hyperref}

\theoremstyle{plain}
\newtheorem{theorem}{Theorem}[section]
\newtheorem{proposition}[theorem]{Proposition}

\newtheorem{corollary}[theorem]{Corollary}
\theoremstyle{definition}
\newtheorem{definition}[theorem]{Definition}
\newtheorem{remark}[theorem]{Remark}

\newcommand{\Pset}{\mathcal{P}}
\newcommand{\Vsp}{\mathcal{V}}
\newcommand{\Zsp}{\mathcal{Z}}
\newcommand{\Esp}{\mathcal{E}}
\newcommand{\Cmeas}{\mathcal{C}}
\newcommand{\WW}{\mathsf{W}}
\newcommand{\Mplus}{\mathcal{M}_{+}}
\newcommand{\occ}{m}
\newcommand{\sig}{\sigma_{\Pset}}
\newcommand{\Pihat}{\widehat{\Pi}}
\newcommand{\Rint}{\mathsf{R}}
\newcommand{\infconv}{\,\square\,}
\newcommand{\PZ}{\mathfrak{P}(\Zsp)}

\title{\bfseries Constitutive-Set Mechanics:\\[2pt] variational mechanics on an admissible set of constitutive laws}
\author{%
Altaf Ahmad Lone\textsuperscript{a}, \
R.~T. Durai Prabhakaran\textsuperscript{a}, \
Tawqeer Nasir Tak\textsuperscript{a},\\[2pt]
Timon Rabczuk\textsuperscript{b,c,$*$}, \
Xiaoying Zhuang\textsuperscript{d}\\[8pt]
\begin{minipage}{0.92\textwidth}\small\itshape
\textsuperscript{a}Department of Mechanical Engineering, Indian Institute of
Technology Jammu, Jagti, NH-44, Jammu, 181221, India\\
\textsuperscript{b}Institute of Computational Mechanics and Artificial
Intelligence, Fudan University, Shanghai, 200433, China\\
\textsuperscript{c}Institute of Structural Mechanics, Bauhaus-Universit\"at
Weimar, Marienstra\ss e 15, Weimar, 99423, Germany\\
\textsuperscript{d}Chair of Computational Science and Simulation Technology,
Institute of Photonics, Leibniz University Hannover, Appelstra\ss e 11A,
Hannover, 30167, Germany\\[4pt]
\upshape\textsuperscript{$*$}Corresponding author.\ \textit{E-mail address:}
\texttt{timon.rabczuk@uni-weimar.de} (T.~Rabczuk).
\end{minipage}}
\date{}

\begin{document}
\maketitle
\vspace{-3em}

\begin{abstract}
\noindent
A structural simulation needs a constitutive law, and experiments rarely
determine one uniquely: several laws may fit the same data and satisfy the same
physical constraints, yet differ where the structure is loaded in ways the
tests never were. Constitutive-Set Mechanics (CSM) keeps all of those laws.
It replaces the single law in the variational formulation by the admissible
set, and lets the mechanics itself decide which members of the set affect the
structural prediction. The framework rests on one observation: a finite
element assembly consults the law only at the strains the structure reaches,
so the incremental energy depends on the law through a weighted record of
those strains, the occupation measure of the state. The energy of a state
under the set is a support function on that measure, and the equilibrium
solve one performs anyway supplies the information needed to bound the worst
case. Equilibrium states become certificates on the robust response, rigorous
on the upper side for any mechanics solver and two-sided under global
minimisation; the certificate has at most one more state than the number of
constitutive directions the structure interrogates; a coherence theorem
identifies when the pointwise-worst material is one no single material can
be; and constitutive inference acts on the same support function. The method
is demonstrated on linear elasticity, a data-constrained constitutive
function, phase-field fracture, an assembled finite-strain composite shell,
and a membrane characterised by one-mode tests alone. On the shell the
structure interrogates four of $226$ constitutive directions; on the membrane
the certified worst-case energy exceeds the reference more than fourfold, and
the experiment that most contracts the certified prediction is not the one
where the constitutive uncertainty is largest.
\end{abstract}

\section{Introduction}\label{sec:intro}

A structural simulation needs a constitutive law. The geometry, the loads,
the kinematics and the discretisation are fixed by the problem; the material
is not, and it has to be supplied. The usual route is calibration: experiments
are performed, a model form is chosen, its parameters are fitted, and the
resulting single law is used in the structural model
\cite{Ogden1972,Fuhg2025}. Where the structure operates near the states the
experiments probed, this works well.

The difficulty is that experiments provide limited information, and several
constitutive laws may be consistent with the same measurements. Suppose
tensile tests characterise a material accurately along each of its axes, but
the structure later undergoes coupled biaxial deformation. Several laws will
agree almost perfectly with the tensile data and differ strongly under the
coupled loading. This is the normal situation for a nonlinear, anisotropic,
multi-mechanism or degrading material: the feasible experiments are largely
restricted to nearly homogeneous load cases, while the structure visits a
coupled multiaxial region that no test reaches \cite{Fuhg2025}. The evidence
then determines not one law but a set of laws, each reproducing the
measurements within their tolerances and satisfying the physics of the
material class. Calibration selects one member of that set because a fitting
procedure happened to select it, and the structural prediction inherits the
selection without any record of the alternatives that were equally plausible.

Constitutive-Set Mechanics (CSM), the framework proposed here, keeps the
admissible set instead of selecting from it, and it asks the mechanics itself
which of the alternatives affect the structural prediction. The aim is not
only to obtain a range of possible responses. It is to determine which part
of the constitutive uncertainty is actually visible to the structure, to
certify what the structure can conclude in spite of it, and to identify which
experiment would reduce the uncertainty that matters for that structure. The
remainder of this introduction places the idea against the work it builds on,
and then states it.

Carrying a set rather than a point is familiar in computational mechanics.
Convex models of uncertainty represent scatter in material properties by a
convex set instead of a distribution \cite{BenHaim1990,Elishakoff1994}, and
anti-optimisation computes the extremes of a structural response over such a
set \cite{Qiu1998}. Non-probabilistic finite element analysis propagates
interval and fuzzy parameter descriptions \cite{Moens2005,Faes2020};
uncertainty quantification and sensitivity analysis measure the influence of
uncertain material inputs, including for phase-field fracture
\cite{Hamdia2015,VuBac2016}; nonparametric probabilistic approaches randomise
the operator itself to represent model-form error \cite{Soize2005}; and
set-membership identification constructs the feasible parameter set consistent
with bounded-error data \cite{Milanese1991}. Robust optimisation computes the
worst case over an uncertainty set by convex duality, worst-case compliance
being the canonical structural instance
\cite{BenTal1997,BenTal1998,Kanno2006,BenTal2009,Bertsimas2011}. Optimal
uncertainty quantification goes furthest, extremising a quantity of interest
over admissible pairs of response functions and probability measures consistent
with a stated information set, and reducing the resulting infinite-dimensional
programmes through the moment structure of that information \cite{Owhadi2013}.
In each of these the object carried is a coefficient vector of a declared model
form or an input--output response map, and the constitutive function that the
variational principle consults is not itself the uncertain input.

Two further lines dispense with calibration, in opposite directions.
Data-driven computational mechanics discards the constitutive equation and
works directly with a material data set
\cite{Kirchdoerfer2016,Kirchdoerfer2017,Conti2018}, with extensions to
inelasticity \cite{Eggersmann2019} and fracture \cite{Carrara2020}, and with
inference performed without a model and without a prior
\cite{Conti2023,Prume2023}. Learned constitutive modelling keeps the
constitutive function and represents it by a network or by a discovered
expression, imposing thermodynamic admissibility, convexity, monotonicity and
material symmetry by construction
\cite{Ghaboussi1991,Amos2017,Linka2021,Masi2021,Klein2022} or discovering the
model form from full-field data \cite{Flaschel2021,Flaschel2022,Linka2023};
networks also serve to approximate the variational solution rather than the
material \cite{Samaniego2020,NguyenThanh2020}. Neither line is confined to
point estimates, since Bayesian formulations return posterior distributions
over constitutive models with quantified uncertainty \cite{Joshi2022}. What
each delivers is a constitutive object for the structure to use, and the
structural consequence follows from using it.

Which measurement to perform next is the subject of optimal experimental
design, classically through information-theoretic utilities on the model
parameters \cite{Lindley1956,Chaloner1995,Huan2013} and, in goal-oriented form,
through utilities on a prediction quantity of interest \cite{Attia2018}, with
active learning in the same role for material representations
\cite{Lookman2019}. These are expectations taken under a prior.

The question taken up here is what the variational mechanics should itself do
when its constitutive input is a set. CSM answers with one replacement in the
formulation,
\[
  \psi\;\longrightarrow\;\Pset ,
\]
the single law giving way to the admissible set. In words: every law that
thermodynamics allows, that agrees with the measurements within their
tolerances, that has the regularity a tangent-stiffness test has certified,
and that respects the known structure of the material, is kept
(Section~\ref{sec:rep}). No element of $\Pset$ is ever selected, and an empty
intersection is not a failure of the method but a refutation of the
constitutive class, a finding that the evidence and the assumed physics are
inconsistent. The replacement does not care how the laws are written down: a
parametric formula, a spline, a Gaussian-process posterior and a neural
network all describe subsets of the space of constitutive functions, and the
method acts on the subset. Nonconvex families enter through their convex hull,
at the cost of a conservative upper certificate and no loss at the lower
endpoint (Proposition~\ref{prop:conv}).

The replacement is workable because of one fact about finite element
assembly that every practitioner knows and rarely uses. A structure does not
consult the constitutive law everywhere; in a given state it evaluates the
law only at the strains reached at its quadrature points, each with a weight.
The record of those strains and weights is the strain-occupation measure
$\occ(z)$ of the state, and the incremental energy is \emph{linear in the
constitutive function} through it: $\Pi(z;\psi)=\int\psi\,d\occ(z)+c(z)$, with
$c(z)$ collecting everything that does not depend on the material. The worst
energy any admissible law can give the state is then the largest value of a
linear functional over a convex set, which convex analysis calls the support
function of $\Pset$, here evaluated on the occupation measure:
$\Pihat(z;\Pset)=\sig(\occ(z))+c(z)$ (Section~\ref{sec:seteng}). One
consequence carries the method. Solving the equilibrium problem for one
candidate material gives, through the occupation measure of the resulting
state, the tangent plane to the worst-case problem at that material. The
equilibrium solves one performs anyway are therefore the oracle of a concave
maximisation over $\Pset$: no a priori sampling or enumeration of the
constitutive space is required, and the mechanics is not reformulated. The convex objects involved are classical, being support
functions and concave value functions \cite{Rockafellar1970,HUL1993}, Sion's
minimax theorem \cite{Sion1958}, and outer linearisation in the sense of Kelley
\cite{Kelley1960}, whose failure to terminate finitely on polytopes is
inherited here.

Three consequences follow, developed in Section~\ref{sec:ccc}. Equilibrium
states become the
supporting cuts of a certificate that is rigorous on its upper side for any
solver and two-sided under global minimisation, whose size is set by the
number of constitutive directions the mechanics interrogates. The conservatism of
pointwise worst-case practice is governed by whether the pointwise envelope is
itself an admissible law; closure of $\Pset$ under pointwise maximum suffices
for that, and some admissibility conditions preserve the closure while others
destroy it. And a measurement acts on the same support function by
inf-convolution, with a guarantee that survives only on the width of the
response interval. The first reduces to classical worst-case compliance for
linear parametric elasticity, as Section~\ref{sec:ex1} exhibits. The second has
analogues in the decomposition of constraint-wise uncertainty in robust
optimisation \cite{BenTal1998} and in the overestimation of coupled ranges by
independent interval evaluation \cite{Moens2005}; Section~\ref{sec:coh} gives
its form for sets of constitutive functions.

The method is stated in Sections~\ref{sec:csm} and~\ref{sec:ccc} and instantiated
on problems of increasing mechanical richness, chosen so that the
constitutive-set layer is tested rather than the mechanics reinvented: linear
elasticity (Section~\ref{sec:ex1}); a constitutive function constrained by
energy and tangent-stiffness data, where the certified regularity makes the
pointwise-worst material inadmissible (Section~\ref{sec:ex2}); isotropic
phase-field fracture \cite{Francfort1998,Bourdin2000,Miehe2010,Ambati2015},
carrying that same set through a nonconvex state space with an irreversible
internal variable (Section~\ref{sec:s3}); an assembled finite-strain composite
shell with mode-resolved degradation
\cite{Amiri2014,Areias2016,Bleyer2018,Bui2021}, where the occupation measure is
a by-product of the production assembly and analytic, spline and learned
representations enter through one socket (Section~\ref{sec:s4}); and a membrane
characterised only by one-mode tests, where the certified worst-case energy
exceeds the reference more than fourfold and a well-fitting network lies
outside the interval the evidence supports (Section~\ref{sec:s5}). CSM certifies value-level quantities of a declared
incremental problem: it does not certify solution fields, crack paths, or the
global optimality of a sequential trajectory, and its path-level statements
hold for trajectory classes independent of the constitutive law.

\subsection{Contributions}\label{sec:contrib}

\begin{enumerate}[label=(\roman*),leftmargin=2em,itemsep=1pt,topsep=2pt,parsep=0pt]
\item \emph{A variational formulation that carries an admissible set of
      constitutive laws rather than selecting one,} through equilibrium,
      degradation and loading history, without modifying the mechanics. It
      rests on one identification: the energy of a state under the set is the
      support function of $\Pset$ on the occupation measure of the state, and
      the occupation measure of an equilibrium state is a supergradient of the
      value function, so the finite element solve is the oracle and the
      supporting functional is a by-product of assembly.
\item \emph{A mechanical certificate generated directly by equilibrium
      states:} a bracket on the worst-case response whose dual certificate
      has at most $n_{\rm vis}+1$ states, $n_{\rm vis}$ being the number of
      constitutive directions the structure interrogates, rather than a size
      set by the discretisation. This bounds the number of states in the
      certificate, not the number of equilibrium solves the adaptive procedure
      needs to construct it; the two are governed by different quantities and
      are kept apart throughout.
\item \emph{A clear distinction between what can be certified with any solver
      and what requires global minimisation:} the upper bound is rigorous for
      any equilibrium solver, the lower bound only under global minimisation,
      and an inversion of the two is a certificate that a local solver has
      returned a non-global state.
\item \emph{A coherence measure that quantifies when pointwise worst-case
      constitutive behaviour cannot be produced by one physical material:}
      the gap is zero exactly when the pointwise envelope is itself an
      admissible law, for which closure of $\Pset$ under pointwise maximum is
      sufficient, together with a classification of which admissibility
      conditions preserve that closure and which couple the law across
      strains.
\item \emph{A mechanics-driven procedure for identifying which experiment
      reduces structural uncertainty:} a measurement acts on the same support
      function that gives the worst case, the worst-case value itself cannot
      be guaranteed to fall, the width of the response interval can, and the
      constitutive uncertainty with the largest structural consequence is
      identified by the mechanics and need not be the largest uncertainty in
      the material description.
\end{enumerate}

\section{Constitutive-Set Mechanics}\label{sec:csm}

This section states the framework. The idea can be given in one paragraph.
A finite element computation minimises an energy over admissible states for a
given material law. That energy depends on the law only through the strains
the structure actually reaches, and it depends on the law \emph{linearly}
through them: the assembly sums the energy density over quadrature points with
fixed nonnegative weights. If the law is replaced by a set of candidate laws,
the worst energy any of them can give a state is therefore a well-understood
object from convex analysis, and the equilibrium solve for one candidate law
already delivers the information needed to bound the worst case over all of
them. The rest of the section makes this precise: the setting, the
representation of constitutive knowledge as a set, the record of visited
strains (the occupation measure), the energy of a state under the set, and the
duality that follows.

One point is worth seeing on the smallest possible instance before the
formalism, because it is the reason a support function is needed at all. Let
the admissible set be the segment between two laws,
$\Pset=\mathrm{conv}\{\psi_1,\psi_2\}$, and let a state consult the law at two
strains with weights $w_a,w_b>0$, so that its occupation measure is
$\occ=w_a\delta_{E_a}+w_b\delta_{E_b}$. The energy of the state under a law
$\psi=\lambda\psi_1+(1-\lambda)\psi_2$ is linear in $\lambda$, so its maximum
over the segment sits at an endpoint:
\[
  \sig(\occ)=\max\bigl\{\,w_a\psi_1(E_a)+w_b\psi_1(E_b),\;
                          w_a\psi_2(E_a)+w_b\psi_2(E_b)\,\bigr\}.
\]
One admissible law serves both strains. The pointwise envelope
$\bar\psi(E)=\max\{\psi_1(E),\psi_2(E)\}$ is free to take $\psi_1$ at $E_a$ and
$\psi_2$ at $E_b$, and if that is where each is larger,
\[
  \int\bar\psi\,d\occ=w_a\psi_1(E_a)+w_b\psi_2(E_b)\;>\;\sig(\occ),
\]
strictly, because each coherent value on the right is smaller in one of its two
terms. The envelope is then not a member of $\Pset$ --- it is not on the
segment --- and the pointwise-worst energy is one that no admissible material
can produce. This is not a theorem or a study; it is the whole of the
coherence question of Section~\ref{sec:coh} on two laws and two strains, and
the reason the energy of a state under a set is defined below through
$\sig$ rather than through $\bar\psi$.

\subsection{Setting}

Variational mechanics selects a state by minimising an incremental functional
$\Pi(z;\psi)$ over admissible states $z\in\Zsp$ for a given constitutive law
$\psi$. The state $z$ collects whatever the problem evolves --- displacements,
and in fracture or plasticity also the internal variables --- and $\Zsp$ is the
set of states the problem allows at the current increment. CSM keeps this
mechanics unchanged and replaces the law by a set: the constitutive input is
$\Pset\subset\Vsp$, a convex compact subset of a space of constitutive
functions on a strain space $\Esp$, and no element of $\Pset$ is ever
selected. Three layers are logically independent and are kept so: the
\emph{mechanics} (any variational problem, with its solver, which we call the
oracle because it answers the question ``what state does material $\psi$
produce?'', $\psi\mapsto z^{\star}(\psi)$), the \emph{constitutive knowledge}
(the set $\Pset$, however constructed), and the \emph{information} (how the
mechanics decides what the knowledge must explain next).
Figure~\ref{fig:workflow} places the objects of the paper on these three
layers, and Table~\ref{tab:notation} lists the recurring ones.

\begin{figure}[t]
\centering
\includegraphics[width=\textwidth]{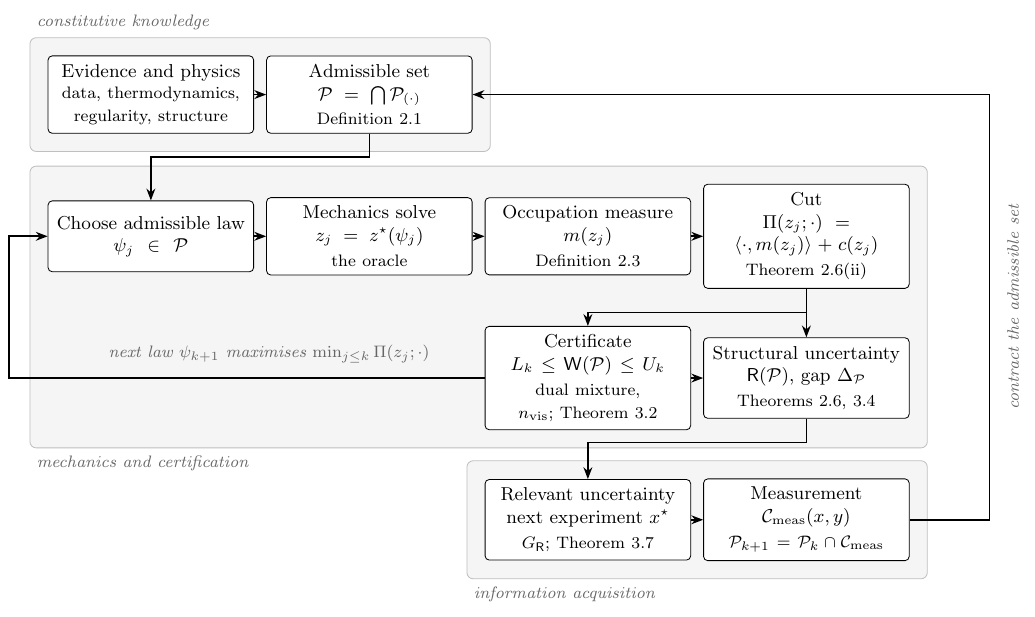}
\caption{The CSM workflow. Evidence and physics define the admissible set
$\Pset$; a mechanics solve at an admissible law returns a state, whose
occupation measure $\occ(z_j)$ supplies a cut on the worst-case problem;
the cuts form a certificate on the robust value, and the response interval is
the structural uncertainty. Two loops close the diagram: the witness loop
chooses the next law from the current cuts, and the information loop
contracts $\Pset$ by a measurement. Every box is an object or an operation
defined in Sections~\ref{sec:csm} and~\ref{sec:ccc}.}
\label{fig:workflow}
\end{figure}

\begin{table}[t]\centering\small
\caption{Recurring objects.}
\label{tab:notation}
\begin{tabular}{@{}l p{0.70\textwidth}@{}}
\toprule
$\Pset\subset\Vsp$ & admissible constitutive set (Definition~\ref{def:P}) \\
$z\in\Zsp$; $z^{\star}(\psi)$ & mechanical state; the state the oracle returns for the law $\psi$ \\
$\occ(z)$ & occupation measure of the state, the weighted record of its visited strains (Definition~\ref{def:occ}) \\
$\Pi(z;\psi)=\int\psi\,d\occ(z)+c(z)$ & incremental energy, affine in $\psi$ \\
$\sig(\occ)$ & support function of $\Pset$ on $\occ$; $\Pihat(z;\Pset)=\sig(\occ(z))+c(z)$ is the set energy \\
$f(\psi)=\inf_z\Pi(z;\psi)$ & mechanical value function, concave in $\psi$ \\
$\WW(\Pset)=\sup_\Pset f$; $\Rint(\Pset)$ & robust value; response interval $[\inf_\Pset f,\WW(\Pset)]$ \\
$L_k,\;U_k$ & lower and upper certificate bounds from $k$ witness states \\
$n_{\rm raw}$; $n_{\rm vis}$ & dimension of the constitutive parameterisation; dimension the mechanics can see \\
$\bar\psi$; $\Delta_\Pset(z)$ & pointwise envelope of $\Pset$; coherence gap of the state \\
$\Cmeas(x,y)$; $G_{\Rint}(x)$ & laws consistent with outcome $y$ of measurement $x$; guaranteed contraction \\
\bottomrule
\end{tabular}
\end{table} The theory below is indifferent to how $\Pset$ is parameterised and to
how the mechanics is discretised; both enter only through the objects of
Section~\ref{sec:seteng}. The paper is methodological: the layer is stated
once, in Sections~\ref{sec:csm} and~\ref{sec:ccc}, and the five studies exercise
it across mechanics of increasing complexity without modifying it.

\subsection{Representation of constitutive knowledge}\label{sec:rep}

What is known about a material comes from four sources, and each of them
excludes some candidate laws rather than selecting one. Thermodynamics
excludes laws that are not stress-free at zero strain, or that are not convex
where convexity is required. Experiments exclude laws that disagree with the
measurements by more than the measurement tolerance. Regularity excludes laws
with unphysically sharp curvature, when a tangent-stiffness test has bounded
the curvature. Material structure excludes laws that violate a known symmetry
or a known relation among moduli. The set of laws that survive all four
exclusions is the constitutive knowledge, and it is what CSM carries.

\begin{definition}[Constitutive knowledge set]\label{def:P}
$\Pset\subset\Vsp$ is the intersection
\[
  \Pset\;=\;\Pset_{\mathrm{thermo}}\cap\Pset_{\mathrm{exp}}\cap\Pset_{\mathrm{reg}}\cap\Pset_{\mathrm{struct}}
\]
of the thermodynamically admissible laws (positivity, normalisation, convexity
or polyconvexity where required, monotonicity of the stress where required),
the laws consistent with experimental data within stated tolerances, the laws
satisfying regularity conditions for the material class (Lipschitz or curvature
bounds on the law or its derivatives), and the laws respecting known material
structure (symmetry, identities among moduli, mixture rules). $\Pset$ is assumed
nonempty, convex and compact in a topology in which the pairing of
Section~\ref{sec:seteng} is continuous. Emptiness of $\Pset$ is refutation of
the constitutive class and is an output of the method, not a failure of it.
\end{definition}

Each component of Definition~\ref{def:P} is a constraint on the constitutive
\emph{function}, and each is either linear or convex in $\psi$; no assumption is
made about how $\psi$ is represented. A parametric law, a spline, a symbolic
form, or a neural network are parameterisations of subsets of $\Vsp$; the
method acts on the subset.

\begin{proposition}[Nonconvex parameterisations]\label{prop:conv}
Let $\Pset_0\subset\Vsp$ be compact but not convex, and let $f$ be any concave
function on $\mathrm{conv}\,\Pset_0$; the case of interest is the mechanical
value function of Section~\ref{sec:value}. Then
\[
  \sup_{\mathrm{conv}\,\Pset_0}f\;\ge\;\sup_{\Pset_0}f,
  \qquad
  \inf_{\mathrm{conv}\,\Pset_0}f\;=\;\inf_{\Pset_0}f .
\]
\end{proposition}

\begin{proof}
The inequality is inclusion. A concave function attains its infimum over a
compact convex set at an extreme point, and every extreme point of
$\mathrm{conv}\,\Pset_0$ lies in $\Pset_0$.
\end{proof}

For the concave value function considered here, a nonconvex family can
therefore be handled through its convex hull: the upper endpoint of the
response interval can be enlarged, so the upper certificate is conservative and
may be strictly so, while the lower endpoint is unchanged. The statement
concerns the two endpoints of the response interval and nothing else;
convexification does not preserve membership of a particular law in the
family, nor any quantity defined on the family rather than on its hull. The
method itself is unchanged, and this is the precise sense in which CSM is
independent of the representation of constitutive knowledge.

\subsection{The occupation measure and the affine pairing}\label{sec:affine}

A structure does not use the constitutive law everywhere in strain space. In a
given state it evaluates the law only at the strains reached at its material
points or quadrature points, and each evaluation enters the total energy with
a weight: the quadrature weight, the metric of the element, and any damage or
degradation factor that scales the density at that point. The object that
records where the structure is asking the law a question, and how strongly
each question counts, is the occupation measure.

\begin{definition}[Occupation measure and pairing]\label{def:occ}
A state $z\in\Zsp$ determines the nonnegative measure
\begin{equation}\label{eq:occ}
  \occ(z)\;=\;\sum_{q}w_q(z)\,\delta_{E_q(z)}\;\in\;\Mplus(\Esp),
\end{equation}
the strains at which the constitutive function is consulted, weighted by
quadrature, metric and any degradation factors. The incremental functional is
\begin{equation}\label{eq:pairing}
  \Pi(z;\psi)\;=\;\int_{\Esp}\psi\,d\occ(z)\;+\;c(z),
\end{equation}
affine in $\psi$ for each fixed $z$, with $c(z)$ collecting the terms
independent of $\psi$.
\end{definition}

Here $\occ(z)$ contains all the strain information of the state that the law
can see, and $c(z)$ contains everything that does not depend on the law: the
work of the external loads, the fracture-surface energy, the penalty and
regularisation terms. Equation~\eqref{eq:pairing} says that once the state is
fixed, the energy is a weighted sum of values of $\psi$ plus a constant. That
is the whole of the structure on which the method rests, and it is worth
being clear about where it comes from.

Equation~\eqref{eq:pairing} is a statement about assembly, not about the
material model. An energy density evaluated at quadrature points and summed
with nonnegative weights is a linear functional of the density, and that
functional is the occupation measure; nothing about the form of $\psi$ enters.
Three features of nonlinear mechanics leave this untouched. Geometric
nonlinearity changes the strains $E_q(z)$ at which $\psi$ is consulted, not
the linearity of the sum. Degradation and damage multiply the weights
$w_q(z)$ by nonnegative factors and so keep $\occ(z)$ in $\Mplus(\Esp)$.
Irreversibility constrains the admissible states --- a monotone phase field, an
obstacle on an interfacial variable --- and never the pairing, so the
incremental functional remains affine in $\psi$ on a state space that is
merely smaller, or nonconvex; the nonconvexity is the subject of
Theorem~\ref{thm:saddle}.

The one operation that destroys affinity is the elimination of an internal
variable by local minimisation inside the constitutive evaluation. The
condensed density is then the minimum over that variable of an affine density,
and is concave rather than affine in $\psi$; the plane-stress condensation of
the through-thickness strain in Section~\ref{sec:s4} is the worked instance.
The repair is exact: the eliminated variable is restored to the state. The
density is affine in $\psi$ on the enlarged state, the occupation measure is
taken there, and every statement below holds on the enlarged state without
change. Concavity on the condensed state would still make tangent
linearisations valid majorants and so preserve the upper certificate of
Section~\ref{sec:cert}, but it would forfeit the exact affine decomposition on
which the finite dual certificate and the coherence theorem rest. Enlargement
is therefore preferred, and it costs only that the constitutive evaluation
expose the variable it eliminates.

The same observation extends from an increment to a loading history. A
trajectory $(z_1,\dots,z_N)$ over $N$ increments is an element of the product
state space, and when the admissible class of trajectories is declared
independently of $\psi$ --- irreversibility $d_k\ge d_{k-1}$ with the increment
problems posed on a fixed history, for instance --- the path functional
$\sum_k\Pi_k(z_k;\psi)$ is affine in $\psi$ with occupation measure
$\sum_k\occ_k(z_k)$, by summation of \eqref{eq:pairing} over increments.
Everything that follows then applies to the path problem verbatim, with the
path solve as the oracle; this is the sense in which Study~3 certifies a
trajectory class. What it does not cover is a class whose admissibility itself
depends on $\psi$, as when the material selects among competing branches; that
case is taken up in Section~\ref{sec:concl}.

\subsection{Set energy and the value function}\label{sec:seteng}\label{sec:value}

With the law replaced by a set, the natural question about a state is: what is
the largest energy any admissible law can give it? Because the energy is a
weighted sum of values of $\psi$, this is the largest value of a linear
functional over the set $\Pset$, which in convex analysis is called the
support function of $\Pset$ in the direction given by the weights. The
direction here is the occupation measure of the state.

\begin{definition}[Set energy]\label{def:seteng}
The energy of a state $z$ under the set $\Pset$ is
\begin{equation}\label{eq:seteng}
  \Pihat(z;\Pset)\;=\;\sup_{\psi\in\Pset}\Pi(z;\psi)\;=\;\sig\bigl(\occ(z)\bigr)+c(z),
  \qquad
  \sig(\occ)=\sup_{\psi\in\Pset}\int\psi\,d\occ .
\end{equation}
\end{definition}

$\sig$ is the support function of $\Pset$, restricted to $\Mplus(\Esp)$: it is
sublinear, finite and lower semicontinuous, and it is not additive in general.
The last property has a physical meaning that the reader should hold on to.
Additivity would mean that the worst energy of a state could be assembled
strain by strain, taking the worst material at each strain separately. It
cannot be, in general, because one material must serve the whole state: the
set energy is nonlocal in strain space. Section~\ref{sec:coh} quantifies
exactly when and by how much.

Three further quantities are needed. For a single material $\psi$, the
mechanics returns the equilibrium energy, which we call the value function
$f(\psi)$. Over the admissible set, the largest such energy is the robust
value $\WW(\Pset)$, the worst case the evidence permits. And the set of all
equilibrium energies admissible materials can produce is the response set.

\begin{definition}[Value function, robust value, response set]
\begin{equation}\label{eq:vals}
  f(\psi)=\inf_{z\in\Zsp}\Pi(z;\psi),\qquad
  \WW(\Pset)=\sup_{\psi\in\Pset}f(\psi),\qquad
  \Rint(\Pset)=f(\Pset)\subset\mathbb{R}.
\end{equation}
\end{definition}

The next result is the reason the method is computable. It says three things
a mechanics reader can use directly: the worst case over the admissible set is
a concave maximisation, so it has no spurious local optima; the response set
is an interval, so ``the range of predictions the evidence permits'' is a
single pair of numbers; and the gradient information needed to find the worst
case is delivered by the finite element solve itself, with no additional
computation.

\begin{theorem}[Structure of the value function]\label{thm:value}
Under \eqref{eq:pairing}:
\begin{enumerate}[label=(\roman*),leftmargin=2em,itemsep=1pt]
\item $f$ is concave and upper semicontinuous, and continuous along every
      segment of $\Pset$; hence $\Rint(\Pset)=[\inf_\Pset f,\ \WW(\Pset)]$ is an
      interval.
\item If the infimum in \eqref{eq:vals} is attained at $z^{\star}(\psi)$, then
      $\occ(z^{\star}(\psi))$ is a supergradient of $f$ at $\psi$, and
      $\Pi(z^{\star}(\psi);\cdot)$ is an affine majorant of $f$ exact at $\psi$.
\item $\WW(\Pset)$ is attained; $\inf_\Pset f$ is attained at an extreme point of
      $\Pset$.
\end{enumerate}
\end{theorem}

\begin{proof}
(i) $f$ is an infimum of affine functions; a finite concave function on a
segment is continuous on the closed segment, and the intermediate value
theorem along segments gives the interval. (ii) $\Pi(z^{\star}(\psi);\cdot)\ge f$
pointwise with equality at $\psi$, which is the definition of a supergradient
of a concave function. (iii) Upper semicontinuity on a compact set; a concave
function attains its infimum over a compact convex set at an extreme point.
\end{proof}

Part (ii) is the mechanical content. A supergradient of a concave function is
the slope of a tangent plane that lies above the function and touches it at
the point in question; it is the analogue of a gradient for a function that
need not be differentiable. Part (ii) says that \emph{the supergradient of the
value function at a material is the occupation measure of the state that
material produces.} In practice: solve the equilibrium problem for one
candidate material, read off the strains at the quadrature points and their
weights, and one has the tangent plane to the worst-case problem at that
material. The oracle that resolves constitutive uncertainty is the mechanics
solve one performs anyway, and the supporting functional it returns is a
by-product of assembly.

\subsection{Mixed-state duality and the complementary principle}\label{sec:duality}

The robust value was defined by taking the worst material and then the
equilibrium state it produces. Can the order be exchanged --- take a state
first and then the worst material for it, and minimise the result over states?
For elasticity the answer is yes, and it is what a worst-case design
computation does. For fracture or plasticity the set of admissible states is
not convex, and the exchange is only legitimate if the states are allowed to
be \emph{mixed}: weighted combinations of several states, in the sense of a
probability distribution over states. A mixed state is not a physical
configuration, but it is the object that makes the exchange exact, and the theorem below says
exactly when a pure state suffices.

\begin{theorem}[Mixed-state duality]\label{thm:saddle}
Let $\Zsp$ be a compact metric space, let $\Pset$ be compact and convex, and
let $\Pi$ be continuous on $\Zsp\times\Pset$. Let $\PZ$ denote the probability
measures on $\Zsp$ with the weak-$*$ topology, in which it is compact and
convex, and for $\mu\in\PZ$ set
$\Pihat(\mu;\Pset)=\sig(\int\occ\,d\mu)+\int c\,d\mu$. Then
\begin{equation}\label{eq:saddle}
  \WW(\Pset)\;=\;\sup_{\psi\in\Pset}\;\inf_{z\in\Zsp}\Pi(z;\psi)
  \;=\;\inf_{\mu\in\PZ}\Pihat(\mu;\Pset),
\end{equation}
both extrema are attained, and a saddle point $(\psi^{\star},\mu^{\star})$
exists. For pure states only weak duality holds:
\begin{equation}\label{eq:weak}
  \inf_{z\in\Zsp}\Pihat(z;\Pset)\;\ge\;\WW(\Pset),
\end{equation}
with equality when $\Zsp$ is convex and $\Pi(\cdot;\psi)$ is convex on $\Zsp$
for every $\psi$.
\end{theorem}

\begin{proof}
$\Pi(\mu;\psi)=\int\Pi(z;\psi)\,d\mu$ is linear and weak-$*$ continuous in
$\mu$ and affine and continuous in $\psi$. For fixed $\psi$,
$\inf_\mu\Pi(\mu;\psi)=\inf_z\Pi(z;\psi)$ because a continuous linear
functional on $\PZ$ is minimised at a Dirac mass, and the infimum over the
compact $\Zsp$ is attained. Sion's minimax theorem applies on the compact
convex $\Pset$ and the convex $\PZ$ and exchanges the order; the inner
supremum over $\Pset$ is $\sig$ evaluated on the mixed occupation measure.
Attainment: $f=\inf_z\Pi(z;\cdot)$ is upper semicontinuous on the compact
$\Pset$, so the supremum on the left of \eqref{eq:saddle} is attained at some
$\psi^{\star}$; $\mu\mapsto\Pihat(\mu;\Pset)$ is a supremum of continuous
functions of $\mu$, hence lower semicontinuous on the weak-$*$ compact $\PZ$,
so the infimum on the right is attained at some $\mu^{\star}$; equality of the
two values then makes $(\psi^{\star},\mu^{\star})$ a saddle point. Without
compactness of $\Zsp$ or continuity of $\Pi$ the equality can hold with one or
both extrema unattained. Weak duality \eqref{eq:weak} is
$\inf\sup\ge\sup\inf$; under the stated convexity Sion's theorem applies to
pure states directly.
\end{proof}

The distinction in Theorem~\ref{thm:saddle} is part of the method. Pure-state
mechanics minimises $\Pi(\cdot;\psi)$ for a given material; robust
certification minimises the set energy over \emph{mixed} states. In elasticity
the two coincide. In fracture, plasticity and any mechanics with a nonconvex
state space they need not, and $\Pihat(z;\Pset)$ minimised over pure states is
then an upper bound on the robust value, not the robust value.

\begin{proposition}[Complementary principle]\label{prop:comp}
Suppose the mechanics admits a complementary variational principle with
functional $\Pi^{c}(\zeta;\psi^{c})$, affine in a dual constitutive variable
$\psi^{c}$ (compliance, dissipation potential) through a dual occupation
measure $\occ^{c}(\zeta)\in\Mplus$, and with value
$f^{c}(\psi^{c})=\inf_\zeta\Pi^{c}=-f(\psi)$ whenever $\psi^{c}$ is the dual of
$\psi$. Let $\Pset^{c}$ be the image of $\Pset$ under this duality. Then
\[
  \inf_{\Pset}f\;=\;-\sup_{\Pset^{c}}f^{c},
\]
so the lower endpoint of $\Rint(\Pset)$ is the upper endpoint of the
complementary problem. If $\Pset^{c}$ is convex, Theorem~\ref{thm:saddle} and
the certificates of Section~\ref{sec:cert} apply to it verbatim; if not,
Proposition~\ref{prop:conv} applies to $\mathrm{conv}\,\Pset^{c}$ and the
resulting certificate of $\inf_\Pset f$ is conservative.
\end{proposition}

The choice of variational principle therefore selects which endpoint of the
response interval the witness machinery certifies. Where both principles exist
and both admissible sets are convex, the whole interval is witness-certified.

\section{Certification, coherence and information acquisition}\label{sec:ccc}

Three practical questions follow from Section~\ref{sec:csm}, and this section
answers them in turn. First, how does one compute the worst case over the
admissible set, and how much of it can be trusted when the finite element
solver is a local one? Second, when is the common engineering shortcut of
taking the worst material at every point actually correct, and when does it
overstate the worst case? Third, once the structure has identified the
constitutive uncertainty that matters, which experiment should be performed
next?

\subsection{Certification by mechanics}\label{sec:cert}

The computation is a loop. Pick an admissible material, solve the equilibrium
problem for it, and record the state. That state gives a tangent plane to the
worst-case problem, an affine function of $\psi$ that lies above the value
function everywhere and touches it at the chosen material; such a plane is
called a cut, and the state that produced it a witness. The minimum of the
cuts collected so far is an upper envelope of the value function, and its
maximum over $\Pset$ is an upper bound on the robust value. The energies of
the recorded states at their own materials are lower bounds, provided the
solver actually found the minimum. Choose the next material where the upper
envelope is highest, and repeat. This is outer linearisation, the cutting-plane
method of convex optimisation, with the finite element solve as its oracle.

One point must be understood before the theorem. A global solver tells us that
no other admissible state has lower energy for the chosen material. A local
solver tells us only that the returned state is an equilibrium. The upper
bound needs only that the returned state be admissible, so it is rigorous
with any solver. The lower bound needs the returned energy to be the true
minimum, so it is rigorous only under global minimisation. CSM treats the two
halves of the bracket differently for this reason, and the studies rely on
the difference.

\begin{definition}[Witnesses and brackets]
Let $F=\{z_1,\dots,z_k\}\subset\Zsp$ be states obtained as minimisers at
materials $\psi_1,\dots,\psi_k\in\Pset$. Define
\begin{equation}\label{eq:bracket}
  L_k=\max_{j\le k}\Pi(z_j;\psi_j),
  \qquad
  U_k=\sup_{\psi\in\Pset}\min_{j\le k}\Pi(z_j;\psi)
     =\min_{\mu\in\mathfrak{P}(F)}\Pihat(\mu;\Pset).
\end{equation}
\end{definition}

\begin{theorem}[Certificates]\label{thm:cert}
\begin{enumerate}[label=(\roman*),leftmargin=2em,itemsep=1pt]
\item \emph{(Upper bound, no hypothesis.)} $\WW(\Pset)\le U_k$ for every
      finite $F$ and every choice of the $\psi_j$; it is required only that
      each $z_j$ be a state.
\item \emph{(Lower bound, under globality.)} If each $z_j$ is a global
      minimiser of $\Pi(\cdot;\psi_j)$, then $L_k\le\WW(\Pset)$.
\item \emph{(Monotonicity.)} $U_k$ is nonincreasing and $L_k$ nondecreasing in
      $k$; the cut functionals do not depend on $\Pset$, so witnesses remain
      valid when $\Pset$ is replaced by a subset.
\item \emph{(Finite dual certificate.)} With $\mathcal{T}=\overline{\mathrm{span}}\{\occ(z):z\in\Zsp\}$
      and $n_{\mathrm{vis}}=\dim\Vsp/\mathcal{T}^{\perp}$, the minimising mixture
      in \eqref{eq:bracket} can be taken with at most $n_{\mathrm{vis}}+1$
      atoms, and at closure the dual mixture $\mu^{\star}$ of
      Theorem~\ref{thm:saddle} is a finite mixture of equilibrium states of the
      same cardinality.
\item \emph{(Two-sided convergence, under globality.)} Suppose that the
      mechanics solve returns a global minimiser of $\Pi(\cdot;\psi)$ at every
      queried $\psi$, so that (ii) holds at every step, and that the cuts are
      equi-Lipschitz on $\Pset$. Then choosing $\psi_{k+1}$ as a maximiser in
      the definition of $U_k$ and solving at it gives $U_k-L_k\to0$, and the
      loop terminates at any prescribed positive tolerance in finitely many
      solves. Without the globality hypothesis $U_k$ still decreases to the
      supremum of the cut minimum, but $L_k$ need not approach $\WW(\Pset)$
      and no two-sided statement holds. Termination at tolerance zero in
      finitely many solves is not guaranteed even under globality.
\end{enumerate}
\end{theorem}

\begin{proof}
(i) Each $\Pi(z_j;\cdot)$ majorises $f$ because $z_j\in\Zsp$; the minimum of
majorants is a majorant, and its supremum over $\Pset$ dominates $\WW$. The
last equality in \eqref{eq:bracket} is linear-programming duality. (ii) When
$z_j$ is a global minimiser, $\Pi(z_j;\psi_j)=f(\psi_j)\le\WW$. (iii) Adding
a majorant cannot raise the minimum; the cuts are functions on $\Vsp$ defined
by states alone. (iv) Every cut functional $\Pi(z_j;\cdot)$ depends on $\psi$ only through
$\langle\psi,\occ(z_j)\rangle$, hence only through the class of $\psi$ in
$\Vsp/\mathcal{T}^{\perp}$, of dimension $n_{\mathrm{vis}}$ by the definition of
$\mathcal{T}$; Helly's theorem then bounds
the number of affine functions needed to determine the maximum of their
minimum. (v) Under globality each solve returns the exact value $f(\psi_j)$ and the
exact supergradient $\occ(z_j)$, and the loop is outer linearisation of a
concave maximisation with an exact first-order oracle, for which convergence
of the gap is standard; without globality the oracle returns a majorant of
$f$ at $\psi_j$ rather than $f(\psi_j)$, which is why the lower sequence loses
its meaning. The failure of exact finite
termination is exhibited by a one-dimensional instance with $\Pset=[-1,2]$
and $f(\theta)=-\theta^{2}$, whose iterates approach but never reach the
maximiser.
\end{proof}

Part (iv) needs its terms explained. The occupation measures of all states
the problem can reach span a subspace of measures, $\mathcal{T}$; a constitutive
direction that is orthogonal to that subspace changes $\psi$ without changing
the energy of any reachable state, so the structure cannot see it. The
visible directions are the remaining ones, and $n_{\mathrm{vis}}$ counts them;
it is a property of the mechanics and the class together, and it is to be
distinguished from $n_{\mathrm{raw}}$, the number of parameters in the
representation of the class, which the mechanics has no reason to see in full.
A dual certificate is a weighted combination of witness states whose set
energy equals the upper bound, so that the bound is not merely asserted but
exhibited by states the mechanics produced. Part (iv) says such a
certificate never needs more states than there are visible directions plus
one. The scope of (iv) must be stated exactly. The \emph{size of the certificate}
--- the number of equilibrium states whose mixture realises the upper bound
--- is governed by the constitutive class and not by the mesh or the number of
degrees of freedom. The \emph{number of solves} required to reach a tolerance
is governed by the geometry of $f$ and is not bounded by
$n_{\mathrm{vis}}$.

\begin{remark}[Local oracles, and diagnostics versus certificates]
The two halves of the bracket are not symmetric, and the asymmetry is used
throughout the studies. If the oracle returns local rather than global
minimisers, (i) is unaffected and $U_k$ remains a valid bound, while the
hypothesis of (ii) fails and $L_k$ becomes a bound relative to the states the
oracle can reach. A local solver certifies only the stationarity or the
algorithmic convergence that it actually tests: neither a small residual nor
convergence of the iteration establishes local minimality, let alone global
minimality, so with such a solver $L_k$ is a solver-realized quantity and not
a rigorous lower bound unless global minimality of the witness state is
established by other means. Algorithmic diagnostics and variational
certificates are therefore distinct objects. An iteration may converge to a
stationary state without establishing that the state is minimising, whereas
the certificate (i) remains valid independently of the convergence properties
of the particular oracle, because it uses the returned state only as an
admissible competitor. An observed inversion $L_k>U_k$ is impossible under
global minimisation and is a certificate that the oracle returned a non-global
state at one of the witness materials.
\end{remark}

Three terms are used with fixed meanings from here on. \emph{Certified} is
reserved for $U_k$, which is rigorous under the assumptions on $\Pset$ alone,
and for $L_k$ and the closed bracket when the globality hypothesis of (ii) is
met, as it is in Studies~1 and~2, whose oracles are exact. When the oracle is
an alternate-minimisation or otherwise local solver, as in Studies~3 and~4,
the value $L_k$ it produces is a \emph{solver-realized} value and is reported
as such; the bracket is then two-sided only conditionally on a globality that
is not established, and its closure is consistent with globality without
proving it. A \emph{guaranteed} contraction of the response interval is the
quantity $G_{\Rint}$ of Section~\ref{sec:info}, which holds for every outcome
of a measurement; the change observed when one outcome is imposed is a
\emph{realized} contraction for that outcome and carries no statement about
others. Certificates concern values of the declared incremental functional.
Reaction--displacement curves and phase fields shown for a reference or an
extreme material are the responses of that one material and are
\emph{illustrative}: certification of an energy carries no statement about a
force or a crack path.

Algorithm~\ref{alg:csm} states the certification loop in terms of the objects
above. Every step is an operation already defined; the algorithm adds nothing
to the theory and is given so that the workflow can be reproduced without
guessing.

\begin{algorithm}[t]
\caption{CSM certification by mechanics.}
\label{alg:csm}
\begin{algorithmic}[1]
\Require admissible set $\Pset$ (Definition~\ref{def:P}); incremental
  functional $\Pi(\cdot;\psi)$ on the admissible states $\Zsp$; mechanics
  oracle $\psi\mapsto z^{\star}(\psi)$; tolerance $\tau>0$
\State choose $\psi_1\in\Pset$; solve $z_1\gets z^{\star}(\psi_1)$; record
  $\occ(z_1)$; form the cut $\Pi(z_1;\cdot)=\langle\cdot,\occ(z_1)\rangle+c(z_1)$
\State $k\gets1$
\Loop
  \State $U_k\gets\sup_{\psi\in\Pset}\min_{j\le k}\Pi(z_j;\psi)$, attained at
    $\psi_{k+1}$ \Comment{upper bound; needs only $z_j\in\Zsp$}
  \State $L_k\gets\max_{j\le k}\Pi(z_j;\psi_j)$
    \Comment{lower bound only if each $z_j$ is a global minimiser; otherwise solver-realized}
  \If{$L_k>U_k$} \State report non-globality of the oracle at some $\psi_j$
    (Remark on local oracles); the certified output is $U_k$ \EndIf
  \If{$U_k-L_k\le\tau$ and the globality hypothesis holds} \State \textbf{break}
    \Comment{two-sided certificate to tolerance $\tau$}
    \ElsIf{the prescribed budget of oracle solves is exhausted} \State \textbf{break}
    \Comment{return $U_k$ as the current certified upper bound; no convergence is asserted} \EndIf
  \State solve $z_{k+1}\gets z^{\star}(\psi_{k+1})$; record $\occ(z_{k+1})$;
    add the cut $\Pi(z_{k+1};\cdot)$; $k\gets k+1$
\EndLoop
\Ensure $U_k$ (certified); $L_k$ (certified only under globality); the
  witness states $z_1,\dots,z_k$; the dual mixture attaining $U_k$
  (Theorem~\ref{thm:cert}(iv)), of at most $n_{\rm vis}+1$ states; the visible
  dimension $n_{\rm vis}$; the inversion diagnostic
\end{algorithmic}
\end{algorithm}

Two remarks on the algorithm. The supremum in line~4 is a finite-dimensional
convex programme on the parameterisation of $\Pset$: a linear programme when
the constraints of Definition~\ref{def:P} are linear in the parameters, as in
every study below. And the number of iterations of the loop is not the number
of states in the dual mixture it returns: the loop may run many times while
the certificate it produces has at most $n_{\rm vis}+1$ states.

\subsection{Constitutive coherence}\label{sec:coh}

Suppose a structure reaches three different strain states. The admissible
material that gives the highest energy at the first strain may not be the one
that gives the highest energy at the second or the third. Taking the maximum
independently at each strain therefore combines pieces of different materials
into a response that no single admissible material can produce. That is
nevertheless what conventional worst-case practice does, and the question is
by how much it overstates the worst case, and when it does not. Let
$\bar\psi(E)=\sup_{\psi\in\Pset}\psi(E)$, the pointwise envelope of the set,
and define the \emph{local worst energy} and the \emph{coherence gap}
\begin{equation}\label{eq:gap}
  \Pihat^{\mathrm{loc}}(z;\Pset)=\int\bar\psi\,d\occ(z)+c(z),
  \qquad
  \Delta_{\Pset}(z)=\Pihat^{\mathrm{loc}}(z;\Pset)-\Pihat(z;\Pset)\;\ge\;0 .
\end{equation}
The first quantity permits the worst material to change from one strain to
another; the second requires one admissible material to explain the whole
state. The gap between them is the cost of coherence, and the theorem says
exactly when it is zero: when the envelope, the function assembled from the
pointwise worst values, is itself an admissible material.

\begin{theorem}[Coherence]\label{thm:coh}
Let $S\subseteq\Esp$. The following are equivalent:
\begin{enumerate}[label=(\roman*),leftmargin=2em,itemsep=1pt]
\item $\sig$ is additive on $\Mplus(S)$;
\item all measures in $\Mplus(S)$ share a common maximiser in $\Pset$;
\item $\bar\psi|_S$ is attained by a member of $\Pset$.
\end{enumerate}
$\Delta_{\Pset}(z)=0$ for every state with $\occ(z)$ supported in $S$ if and
only if (i)--(iii) hold.
\end{theorem}

\begin{proof}
$\sup\int\psi\,d\occ\le\int\sup\psi\,d\occ$ gives $\Delta_\Pset\ge0$. (iii)
$\Rightarrow$ (ii) $\Rightarrow$ (i) is immediate. For (i) $\Rightarrow$ (iii):
additivity on finite nonnegative combinations of Dirac masses on $S$ forces a
common maximiser of $\psi(E)$ over each finite subset of $S$, and compactness
of $\Pset$ extends this to $S$.
\end{proof}

\begin{corollary}[Which physics couples strains]\label{cor:lattice}
If $\Pset$ is closed under pointwise maximum, then $\bar\psi\in\Pset$ and
$\Delta_{\Pset}\equiv0$; closure is a sufficient condition for coherence, not
a necessary one, the necessary and sufficient condition being (iii) of
Theorem~\ref{thm:coh}. Pointwise bounds, positivity, monotonicity, convexity
and a uniform Lipschitz bound are all preserved by pointwise maximum, so a set
defined by any combination of these is closed under it. Differentiability,
bounds on the second derivative (curvature or tangent-stiffness bounds),
moduli of continuity on the derivative, identities among moduli, and any
other constraint linking the law at different strains are not preserved; for
such sets the envelope $\bar\psi$ is in general not admissible, and these are
the conditions under which $\Delta_{\Pset}>0$ can occur.
\end{corollary}

\begin{proof}
The maximum of two convex functions is convex, of two monotone functions
monotone, of two functions bounded by $b$ bounded by $b$, and of two
$L$-Lipschitz functions $L$-Lipschitz, since
$|\max(a,b)-\max(a',b')|\le\max(|a-a'|,|b-b'|)$. The maximum of two
differentiable functions has a kink where they cross with distinct slopes, so
differentiability, any bound on the second derivative and any modulus of
continuity on the first derivative fail there; the maximum of two functions
satisfying a linear identity among their moduli does not in general satisfy
it.
\end{proof}

The physical content is this. The gap appears when the constitutive knowledge
couples the law at different strains, so that being the worst at one strain
constrains how bad the same material can be at another. Bounds and shape
conditions do not couple strains, and the pointwise-worst practice is then
exact. A curvature bound couples strains, because a law that is worst at two
neighbouring strains would have to bend too sharply between them, and the
pointwise-worst practice is then conservative. The coupling among strains that
makes the coherent worst case differ from the pointwise worst case is
therefore not assumed; it is a consequence of the
regularity and structural components of Definition~\ref{def:P}. The
pointwise-worst practice is exact when constitutive knowledge consists only of
bounds and shape conditions, and becomes conservative by exactly
$\Delta_{\Pset}$ once the smoothness class or the structure of the material is
known.

\subsection{Mechanics-driven information acquisition}\label{sec:info}

Once the structural problem has identified the constitutive uncertainty that
matters, the next question is which experiment to perform. Not every
experiment is useful. An experiment may probe a region of strain space where
the admissible laws disagree strongly but which the structure never visits;
it then contracts the admissible set without changing anything the structure
can see. The right measure of an experiment is not how much constitutive
uncertainty it removes but how much of the structurally relevant uncertainty
it removes, and this subsection makes that measure precise.

A measurement $x$ with outcome $y$ constrains the law to the set
$\Cmeas(x,y)$ of laws that would have produced that outcome within the
measurement tolerance. Throughout, $\Cmeas(x,y)$ is assumed closed and convex
in $\Vsp$. This is an assumption on the measurement model, not a property of
arbitrary experiments; it holds for the bounded-error observation models used
in the studies, in which the observed quantity --- an energy, a stress or a
tangent modulus at a prescribed strain, or a modulus itself --- is affine in
$\psi$ and the tolerance is an interval, so that $\Cmeas(x,y)$ is a slab.
Constitutive knowledge is then updated by contraction,
\begin{equation}\label{eq:update}
  \Pset_{k+1}=\Pset_k\cap\Cmeas(x,y),
  \qquad
  \sigma_{\Pset_{k+1}}=\overline{\sigma_{\Pset_k}\infconv\sigma_{\Cmeas(x,y)}} .
\end{equation}
The second identity is the classical inf-convolution rule for support
functions of intersections; the inf-convolution of two functions at a point
is the smallest value obtainable by splitting the point between them. Its
mechanical reading is that a measurement acts on the set energy of a state by
splitting the occupation measure between prior knowledge and the measurement,
each bounding the part on which it is tighter. The same support function that
gives the worst-case energy of a state is the object that a measurement
updates: constitutive inference and mechanics operate on the same object.

Two questions are then natural, and the answers differ. Can a measurement be
guaranteed to lower the worst-case prediction? No, for a reason that is
obvious once stated: the worst admissible material may be the true one, and
no measurement can exclude the truth. Can a measurement be guaranteed to
narrow the range of predictions the evidence permits? Yes, and the condition
is that the measurement distinguishes the two extreme admissible materials
from each other by more than its tolerance.

\begin{proposition}[The one-sided guarantee is vacuous]\label{prop:vacuous}
Let the outcome of measurement $x$ be a function $Y_x$ of the material,
possibly with tolerance, so that $\psi\in\Cmeas(x,Y_x(\psi))$ for every
admissible $\psi$, and let the outcome set $Y(x)$ contain $Y_x(\Pset)$. Then
\[
  \WW(\Pset)-\max_{y\in Y(x)}\WW\bigl(\Pset\cap\Cmeas(x,y)\bigr)\;=\;0
  \qquad\text{for every }x .
\]
\end{proposition}

\begin{proof}
Let $\psi^{\star}$ attain $\WW(\Pset)$ and $y^{\star}=Y_x(\psi^{\star})$. Then
$\psi^{\star}\in\Pset\cap\Cmeas(x,y^{\star})$, so the robust value of the
intersection is at least $f(\psi^{\star})=\WW(\Pset)$; inclusion gives the
reverse inequality.
\end{proof}

The worst admissible material may be the true one, and no measurement
consistent with the admissible set is guaranteed to exclude it. A guarantee
must therefore be sought on a quantity that no single outcome can preserve.

\begin{theorem}[Guaranteed contraction of the response interval]\label{thm:gain}
Let $|\Rint(\Pset)|=\WW(\Pset)-\inf_\Pset f$ be the width of the response
interval of \eqref{eq:vals}, and
\[
  G_{\Rint}(x)=|\Rint(\Pset)|-\sup_{y\in Y(x)}\bigl|\Rint\bigl(\Pset\cap\Cmeas(x,y)\bigr)\bigr| .
\]
Then $G_{\Rint}(x)\ge0$, and $G_{\Rint}(x)>0$ only if no outcome set
$\Cmeas(x,y)$ contains both a maximiser and a minimiser of $f$ on $\Pset$:
a measurement has guaranteed value only if it discriminates between the
extreme admissible materials beyond its tolerance. Conversely, if the extreme
materials are unique and the supremum in the definition of $G_{\Rint}$ is
attained --- as it is when $Y(x)$ is finite, or compact with
$y\mapsto|\Rint(\Pset\cap\Cmeas(x,y))|$ upper semicontinuous --- then this
condition is also sufficient.
\end{theorem}

\begin{proof}
Both endpoints are monotone under inclusion, so $|\Rint|$ is. If some
$\Cmeas(x,y)$ contains a maximiser $\psi^{\star}$ and a minimiser $\psi_{\star}$,
the intersection attains both endpoints and $|\Rint|$ is unchanged. For the
converse, let the supremum be attained at $y_0$. By hypothesis
$\Pset\cap\Cmeas(x,y_0)$ excludes one of the two extreme materials, say
$\psi^{\star}$. The upper endpoint of the intersection is attained on the
compact set $\Pset\cap\Cmeas(x,y_0)$ by Theorem~\ref{thm:value}(iii) applied to
it, at a material that is not $\psi^{\star}$; uniqueness of the maximiser then
gives a strictly smaller value, so the width at $y_0$ is strictly smaller than
$|\Rint(\Pset)|$, and this width is the supremum. Attainment of the supremum is
what the argument uses: without it, a sequence of outcomes each excluding
$\psi^{\star}$ by less and less can drive the width back to $|\Rint(\Pset)|$
without any single outcome achieving it.
\end{proof}

The acquisition rule is
\begin{equation}\label{eq:acq}
  x^{\star}\in\arg\max_x G_{\Rint}(x),
\end{equation}
and it is mechanical in two ways: the extreme materials $\psi^{\star}$,
$\psi_{\star}$ are determined by the mechanics through the occupation measures
of the states it visits, and the discrimination that decides the guarantee is
the difference between what those two materials predict at the measurement.
Computing $G_{\Rint}$ exactly requires both endpoints of the interval on each
candidate intersection; the upper endpoint is witness-certified, and the lower
endpoint is the concave minimisation of Theorem~\ref{thm:value}(iii),
enumerable or bounded in the visible quotient. The cost of guaranteed
acquisition is the cost of the lower endpoint, and no cheaper certificate of a
positive gain follows from the witnesses alone. Two quantities must therefore
be kept apart. $G_{\Rint}(x)$ is guaranteed: it is the contraction that every
outcome of $x$ produces at least. The contraction obtained when one particular
outcome is imposed --- the outcome of a hidden material, as in Studies~3
and~5 --- is realized for that outcome only; it may exceed $G_{\Rint}(x)$ by
any amount while $G_{\Rint}(x)$ itself is zero. In plain terms: the
guaranteed contraction is what one can promise before the experiment is
done; the realized contraction is what one observes afterwards, for the
outcome that actually occurred.

\section{Numerical studies}\label{sec:studies}

The five studies exercise the layer of Sections~\ref{sec:csm} and~\ref{sec:ccc}
on mechanics of increasing complexity. Nothing in the layer changes between
them; what changes is the oracle, the state space and the constitutive class.
Each study is introduced by the question it answers and the reason the
example was chosen, and each closes with what is and is not being claimed.

\subsection{Study 1: linear elasticity}\label{sec:ex1}

\emph{The question.} When the constitutive uncertainty is a set of elastic
moduli, does CSM reproduce the classical worst-case compliance computation,
and what does it add? This is the case in which the framework reduces to a
known construction, and it is included to exhibit that reduction and to
introduce the objects of Sections~\ref{sec:csm} and~\ref{sec:ccc} on the
simplest possible mechanics, where every quantity can be checked by hand. The
mechanics is a planar truss, statically indeterminate to degree two, with
thirteen members of unit area, pinned at one support and on a roller at the
other, loaded at the top-middle node. Two members meet at an unloaded free node
and carry no force under any material: the loading programme does not
interrogate them. The constitutive knowledge set is
\[
  \Pset=\bigl\{E\in\mathbb{R}^{13}:\ E_{\mathrm{lo}}\le E\le E_{\mathrm{hi}},\
  \bigl|\tfrac{1}{13}\textstyle\sum_e E_e-1\bigr|\le0.03\bigr\},
\]
per-member bounds $[0.8,1.2]$ (one loaded member $[0.75,1.25]$, the two
uninterrogated members $[0.6,1.4]$) together with a batch-mean certification: the
members come from one batch whose mean modulus is known to $\pm3\%$. The
batch constraint is an instance of $\Pset_{\mathrm{struct}}$ in
Definition~\ref{def:P}, and it is the only coupling among members.

\paragraph{Pairing.}
$\Pi(u;E)=\tfrac12u^{\!\top}K(E)u-f^{\!\top}u$ with $K(E)=\sum_e E_eK_e$; the
occupation measure is the vector of member strain energies per unit modulus,
$\occ_e(u)=\tfrac12u^{\!\top}K_eu\ge0$. The value function is
$f(E)=-\tfrac12 c(E)$ with $c$ the compliance, so
\[
  \Rint(\Pset)=\bigl[-\tfrac12\max_\Pset c,\ -\tfrac12\min_\Pset c\bigr]:
\]
the witness-certified endpoint $\WW(\Pset)$ is the \emph{stiffest} coherent
material, and the classical worst-case compliance $\max_\Pset c$ is the lower
endpoint, a concave minimisation. By Proposition~\ref{prop:comp}, the
complementary principle --- member forces $s$ with $B^{\!\top}s=f$ and
$\Pi^{c}(s;C)=\tfrac12\sum_eC_e(L_e/A_e)s_e^{2}$ in the compliances $C_e=1/E_e$
--- has value $f^{c}=\tfrac12c$ and stress-occupation measure
$\occ^{c}_e(s)=\tfrac12(L_e/A_e)s_e^{2}$, so its witness-certified endpoint is
the classical worst-case compliance directly.

\paragraph{Results.}

\begin{center}\footnotesize
\begin{tabular}{@{}lr@{}}
\toprule
$\WW(\Pset)$, witness loop; solves to relative width $10^{-3}/10^{-5}/10^{-7}$ & $-1.60366$;\ 4 / 8 / 11 \\
\quad active witnesses in the dual certificate; $n_{\mathrm{vis}}$ of 13 & 2;\ 8 \\
\quad uninterrogated members identified as invisible; value change under their motion & yes;\ $1.3\times10^{-15}$ \\
$\inf_\Pset f$, vertex enumeration in the visible quotient (2904 candidates) & $-2.30621$ \\
\quad classical worst-case compliance $\max_\Pset c=-2\inf_\Pset f$ & $4.61242$ \\
\quad 22\,119 admissible random samples all inside the certified interval & yes \\
\midrule
Coherence gap, stiff end: $f(E_{\mathrm{hi}})-\WW(\Pset)$, $E_{\mathrm{hi}}\notin\Pset$ & $4.26\times10^{-2}$\ (2.66\%) \\
Coherence gap, soft end: $c(E_{\mathrm{lo}})-\max_\Pset c$, $E_{\mathrm{lo}}\notin\Pset$ & $1.08\times10^{-1}$\ (2.33\%) \\
Box-only set (lattice): $\WW=f(E_{\mathrm{hi}})$ exactly, 2 solves & yes \\
\midrule
Complementary principle, box-only set: $\sup f^{c}=\tfrac12c(E_{\mathrm{lo}})$, 2 solves, width 0 & yes \\
Complementary principle, coupled set, hull of mapped vertices: $\sup f^{c}$ vs $\tfrac12\max_\Pset c$ & equal to $10^{-15}$ \\
\midrule
Value-criterion gain $G_{\WW}(e)$, every member & $0$ \\
Interval-criterion gain $G_{\Rint}(e)$: best member (most loaded) & member 10:\ $0.121$ \\
\quad member chosen by parameter width (uninterrogated) & member 12:\ $0.060$ \\
\quad $G_{\Rint}>0$ only where $|E^{\star}_e-E_{\star,e}|>2\cdot$tolerance & yes \\
\bottomrule
\end{tabular}
\end{center}

\begin{figure}[t]\centering
\includegraphics[width=\textwidth]{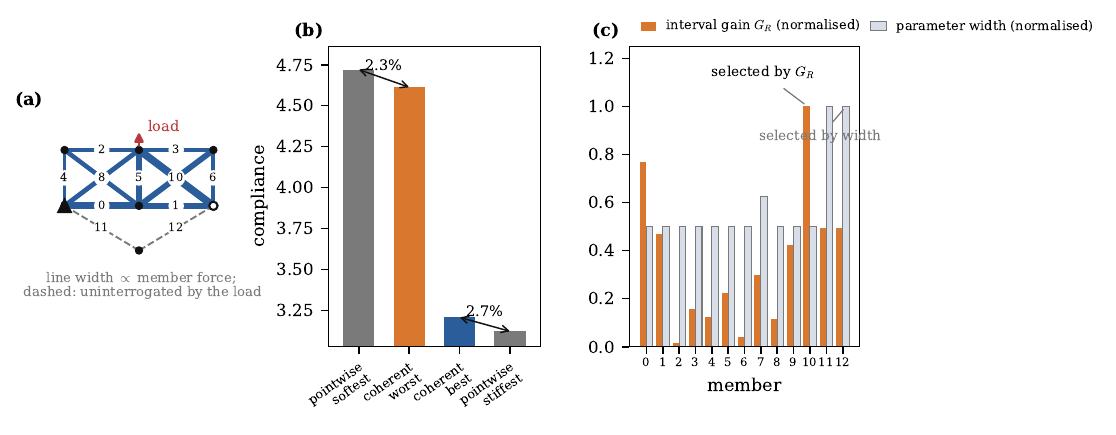}
\caption{Linear elasticity. (a) The truss; line width is proportional to the
member force under the load, and the two members meeting at the unloaded node
carry none. (b) Compliance bounds: the pointwise-softest and pointwise-stiffest
materials lie outside the admissible set, and the coherent worst and best
cases are tighter by the coherence gaps at the two ends. (c) Guaranteed
contraction of the response interval per candidate modulus measurement, against
the parameter width that a model-agnostic criterion would use.}
\label{fig:ex1}
\end{figure}

\paragraph{Reading.}
With per-member bounds alone the set is a lattice, the worst material is the
pointwise-worst material, and CSM reproduces the classical bound in two solves
with no gap. The batch-mean certification is one physically meaningful
constraint, and it is enough to make the pointwise-worst material inadmissible
at both ends: the classical worst-case compliance overstates the coherent worst
case by $2.3\%$, and the stiffest coherent material lies $2.7\%$ below the
pointwise-stiffest one. These are the coherence gaps of
Theorem~\ref{thm:coh} (Figure~\ref{fig:ex1}b), produced by Corollary~\ref{cor:lattice}'s structural
component and not by an assumed budget. The complementary principle certifies
the classical worst-case compliance in two solves for the box, and on the hull
of the mapped vertices for the coupled set, where in this instance the hull
certificate is exact. The uninterrogated members are invisible to the
mechanics, yet the interval gain of measuring them is positive because the
batch constraint makes them budget partners of the loaded members; the gain of
measuring the most heavily loaded member is twice as large, and the parameter
width --- the criterion that ignores the mechanics --- selects the wrong
member. The value-criterion gain is identically zero, as
Proposition~\ref{prop:vacuous} requires.

\subsection{Study 2: a constitutive function constrained by data}\label{sec:ex2}

\emph{The question.} When the unknown is a constitutive \emph{function}
rather than a vector of moduli, does knowing something about its smoothness
change the worst-case prediction? This is the study in which the regularity
component of Definition~\ref{def:P} can act, and the example is chosen so that
the mechanics is exact and the effect of a single certified curvature bound
can be isolated. The mechanics is a tapered bar of five segments of equal length, with
cross-sections chosen so that a single equilibrium state visits the five
distinct strains $S=\{0.034,0.038,0.044,0.052,0.065\}$. The constitutive input
is a nonlinear stored energy $\psi(\varepsilon)$, represented by its nodal
values $\psi_i$ on a strain grid of spacing $h=10^{-3}$ over the operating
interval $[0,0.08]$, with piecewise-linear interpolation between knots; the
law is convex and piecewise linear, so the finite-dimensional mechanics
problem is a linear programme, solved to the stated optimisation tolerances,
and the oracle is global. Derivatives of the law are
accordingly divided differences on the knot grid: $\psi'$ is the slope of each
segment, and $\psi''$ denotes throughout the second divided difference
$D_h^{2}\psi_i=(\psi_{i-1}-2\psi_i+\psi_{i+1})/h^{2}$ at each interior knot,
the jump in slope across the knot divided by $h$. Convexity of the
piecewise-linear law is $D_h^{2}\psi\ge0$ at every interior knot, and a
curvature bound $\psi''\le\kappa$ is $D_h^{2}\psi\le\kappa$ there, the discrete
counterpart of the bound for the smooth law it approximates. The
occupation measure of a state is $\occ(u)=\sum_eA_eL_e\,\delta_{\varepsilon_e(u)}$.
The closure that Corollary~\ref{cor:lattice} requires is verified on the
representation actually used. The visited strains $S$ and the
energy-observation strains are knots of the grid, and every constraint
defining $\Pset_A$ --- the normalisation, nonnegative first and second divided
differences, and bounds on nodal values --- is preserved by the knot-wise
maximum of two admissible nodal vectors. $\Pset_A$ is therefore closed under
that maximum, the envelope on $S$ is attained by a member of $\Pset_A$, and
Theorem~\ref{thm:coh}(iii) applies on $S$. This is closure of the nodal class
under the knot-wise maximum, which is what the theorem requires on $S$;
closure of the fixed-grid class under the pointwise maximum in function
space, which can place a breakpoint between knots, is neither needed nor
asserted.

\paragraph{Constitutive knowledge.}
The set is built from a stated synthetic measurement model on a reference law
$\psi_0=50\varepsilon^{2}+400\varepsilon^{3}$: energy observations at seven
strains with per-test relative uncertainties between $3\%$ and $12\%$, and
tangent-stiffness observations of $E_t=\psi_0''$ at five strains with $4\%$
Gaussian noise, which enter only through the envelope $\kappa_z$ below. Then
\[
  \Pset_A=\Pset_{\mathrm{exp}}\cap\Pset_{\mathrm{thermo}},
  \quad
  \Pset_B=\Pset_A\cap\{D_h^{2}\psi\le\kappa_z\ \text{at every interior knot}\},
  \quad
  \kappa_z=\max_i\bigl(\hat E_{t,i}+z\,\sigma_i\bigr),
\]
with $\Pset_{\mathrm{thermo}}=\{\psi_0=0,\ \psi'\ge0,\ D_h^{2}\psi\ge0\}$ and
$\kappa_z$ the upper confidence envelope of the measured tangent modulus at
level $z$ ($95\%$, $99\%$, $99.9\%$). The regularity constraint is therefore
certified from the tangent-stiffness experiment and is part of the
constitutive knowledge; it is not a coupling chosen to produce an effect. The
reference law is admissible in every case.

\paragraph{Results.}

\begin{center}\small
\begin{tabular}{@{}lr@{}}
\toprule
Case A, $\Pset_A$ (lattice-closed): coherence gap $\Delta_{\Pset}$ at the reference state & $0$ exactly \\
Case B, $\kappa_{95}=304$: $\Delta_{\Pset}$ (of the set energy); binding curvature rows & $1.97\%$;\ 35 \\
Case B, $\kappa_{99}=312$ & $2.01\%$ \\
Case B, $\kappa_{99.9}=320$ & $2.03\%$ \\
$\Delta_{\Pset}$ at $1.5\kappa_{95}$ / $3\kappa_{95}$ / $10\kappa_{95}$ / $50\kappa_{95}$ & $2.15\%$ / $1.16\%$ / $0.18\%$ / $0$ \\
Identity $\Delta_{\Pset}=\int\bar\psi\,d\occ-\sig(\occ)$ & to $3\times10^{-17}$ \\
\midrule
$\WW(\Pset_B)$, witness loop: width, solves, active witnesses, $n_{\mathrm{vis}}$ of 81 & $4\times10^{-15}$, 12, 2, 12 \\
$\WW(\Pset_A)\ge\WW(\Pset_B)\ge f(\psi_0)$ & $-0.2904\ge-0.2927\ge-0.3025$ \\
\bottomrule
\end{tabular}
\end{center}

\begin{figure}[t]\centering
\includegraphics[width=\textwidth]{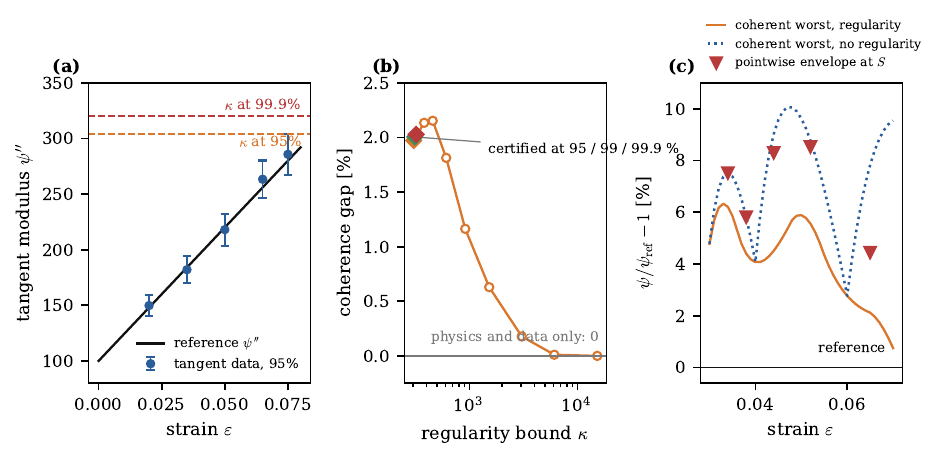}
\caption{Constitutive function constrained by data. (a) Tangent-stiffness
observations with their $95\%$ intervals, the reference tangent modulus, and
the certified regularity envelopes. (b) Coherence gap at the reference state
as the regularity bound is relaxed from the certified levels to infinity.
(c) Deviation from the reference law of the coherent worst material with and
without the certified regularity, and of the pointwise envelope at the five
strains the bar visits.}
\label{fig:ex2}
\end{figure}

\paragraph{Reading.}
With energy data and thermodynamic admissibility alone the pointwise envelope
$\bar\psi$ is itself an admissible law, and the coherent worst case equals the
pointwise worst case exactly, as Corollary~\ref{cor:lattice} requires. Adding
the certified tangent-modulus envelope changes nothing pointwise --- the
reference law and every band remain as they were --- and makes the pointwise
envelope inadmissible: the coherent worst material is a maximal-curvature law
bridging the tight data points, and it lies $2\%$ below the envelope on the
strains the bar visits (Figure~\ref{fig:ex2}c). The effect is insensitive to the confidence level at
which the regularity is certified (Figure~\ref{fig:ex2}b), and it disappears only when the admissible
tangent modulus is allowed to exceed the measured envelope by an order of
magnitude, that is, when the tangent-stiffness information is discarded. This
is the phenomenon of Theorem~\ref{thm:coh} arising from a defensible
constitutive knowledge set, and it is what Study~1, with its
lattice-closed parameter bounds, cannot exhibit except through structural
coupling. The witness certificate on $\Pset_B$ closes in twelve solves with
two active witnesses in a visible quotient of dimension twelve out of
eighty-one nodal parameters.

The nodal representation is one parameterisation of $\Pset_{\mathrm{exp}}$;
a spline, a polynomial basis, a Gaussian-process posterior or a neural
constitutive family parameterise the same knowledge, with
Proposition~\ref{prop:conv} governing the nonconvex cases, and nothing in
Sections~\ref{sec:csm} and~\ref{sec:ccc} changes.

\subsection{Study 3: isotropic phase-field fracture}\label{sec:s3}

\emph{The question.} Can the same constitutive-set framework still provide
useful certificates when the state space is nonconvex, damage evolves
irreversibly, and the strains the structure visits change from increment to
increment as a crack grows? Phase-field fracture is chosen because it has all
three features at once, and because its standard solver, alternate
minimisation, is a local one, so that the distinction between the two halves
of the bracket is tested rather than assumed. The study distinguishes three
levels that the theory keeps separate:
incremental certification (Theorem~\ref{thm:cert}, exact), trajectory
certification on a declared $\psi$-independent class (Section~\ref{sec:affine},
exact), and certification of the trajectory a sequential local solver produces
(not within the theory).

\paragraph{Modelling assumptions.}
AT2 phase-field fracture on a unit square, $12\times12$ structured triangles,
$\ell=2h$, $G_c=0.04$, residual stiffness $10^{-4}$, an initial crack modelled
by $d=1$ on the mid-row up to half width. Energy density $g(d)\,\psi(r)$ with
$r=\|\varepsilon\|_F$ and $g=(1-d)^2+k_{\rm res}$; no tensile--compressive split
(monotone tension). Displacement control, ten monotone increments
$\lambda_1<\dots<\lambda_{10}$ spanning elastic loading, softening, crack
propagation and post-fracture. State $z_k=(u_k,d_k)$ with $d_k\ge d_{k-1}$.
The constitutive set is $\Pset_B$ of Study~2, \emph{inherited unchanged}
(254 rows on 81 nodal values, tangent regularity certified at $95\%$), extended
beyond the last knot by its last slope. The occupation measure is
$\occ_k=\sum_eA_e\,g(d_e)\,[\text{hat weights at }r_e]$, nonnegative and linear
in the nodal values, so $\Pi_k(z;\psi)=\langle\psi,\occ_k(z)\rangle+\Pi_{\rm frac}(d)$
without modification of the CSM layer. For a fixed material, the displacement
subproblem is a second-order cone programme (exact for the piecewise-linear
law), the damage subproblem a convex quadratic programme with the obstacle, and
alternate minimisation from $d_{k-1}$ is the sequential local solver.
Increment problems are posed on the damage reached by the reference law at the
previous increment.

\begin{figure}[t]\centering
\includegraphics[width=\textwidth]{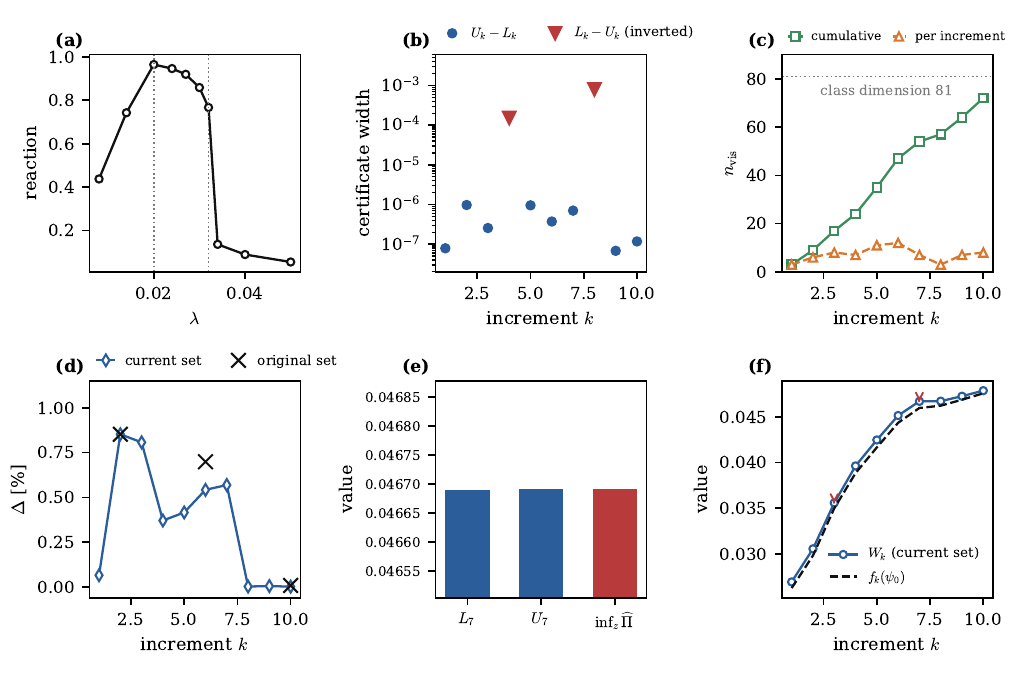}
\caption{Study 3. (a) Reaction versus prescribed displacement for the reference
law (illustrative); dotted lines mark the information updates. (b) Bracket
width $U_k-L_k$ per increment, $L_k$ solver-realized; inverted brackets at
$k=4$ and $k=8$ are detections of a non-global state returned by the local
solver. (c) Mechanically visible constitutive
dimension, per increment and cumulative, against the class dimension.
(d) Coherence gap at the reference state of each increment, on the current set
and on the original set. (e) Mixed-state bracket and pure-state minimum of
the set energy at $k=7$. (f) Upper certificate $U_k$ along the loading history
with the two set contractions.}
\label{fig:s3}
\end{figure}

\begin{figure}[t]\centering
\includegraphics[width=0.9\textwidth]{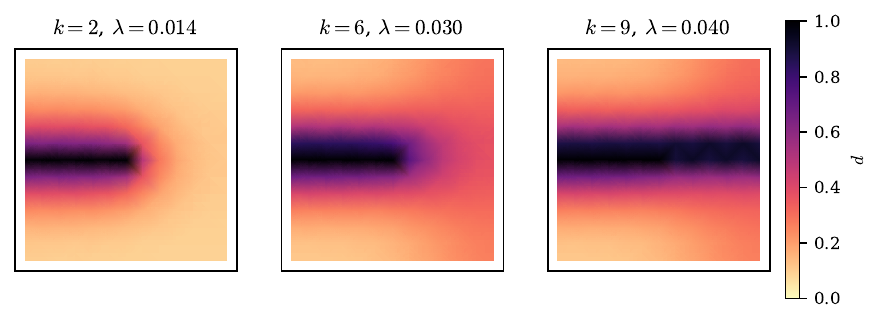}
\caption{Phase field of the reference law before initiation, during
propagation, and after fracture; illustrative, not a certified crack path.}
\label{fig:s3d}
\end{figure}

\Needspace{16\baselineskip}
\paragraph{Results.}

\begin{center}\footnotesize
\begin{tabular}{@{}p{0.66\textwidth}r@{}}
\toprule
Increments; mechanics solves for incremental certification (18 more for the two updates) & 10;\ 72 \\
Bracket width $U_k-L_k$ at closed increments, $L_k$ solver-realized (tolerance $10^{-6}$) & $10^{-7}$--$10^{-6}$ \\
Active witnesses per increment; $n_{\rm vis}$ per increment; cumulative at $k=10$ & 1--4;\ 3--12;\ 72 of 81 \\
Inverted brackets, $k=4$ and $k=8$ (non-global states of the local solver) & $1.5\times10^{-4}$, $8.0\times10^{-4}$ \\
Coherence gap, original set: before initiation $k=2$ / initiation $k=6$ / after $k=10$ & $0.85\%$ / $0.70\%$ / $0.01\%$ \\
Pure-state minimum of the set energy against the bracket at $k=7$, $V_{\rm pure}-U_7$ & $-5\times10^{-7}$ \\
Information updates $\mathsf W_3$, $\mathsf W_7$ after one energy measurement each & $-0.13\%$, $-0.12\%$ \\
Trajectory lift on the declared class: width, path solves, active witnesses & $3\times10^{-7}$, 8, 2 \\
Sequential minus declared-class trajectory value over the eight materials visited & $7\times10^{-6}$ to $4\times10^{-5}$ \\
\bottomrule
\end{tabular}
\end{center}

\paragraph{Reading.}
\emph{Certification.} The same witness loop, support function and
occupation measure that certified the truss and the bar furnish an upper
certificate $U_k$ at every increment of a propagating crack. Because the
oracle is alternate minimisation, $L_k$ is solver-realized: at the increments
where no inversion occurred the closed bracket is consistent with globality
without establishing it. The certificates have a cardinality that never
exceeds four in a visible quotient that grows from 3 to 72 of 81 as the crack
develops and new strains are visited (Figure~\ref{fig:s3}c). The interrogated
dimension is shaped by the mechanical history and plateaus once the crack is
through.

\emph{The certificate as a detector.} At $k=4$ (near peak load) and $k=8$ (crack-through)
the bracket inverted: $L_k>U_k$. Since $U_k$ is valid for any oracle by
Theorem~\ref{thm:cert}(i), an inversion is a certificate that the sequential
local solver returned a non-global state at one of the witness materials. The
theory predicted that under a local oracle the upper bound survives and the
lower bound becomes reachability-relative; the fracture problem exhibited this
at exactly the two increments where the energy landscape is most nonconvex,
and the certificate detected it without knowledge of the true minimum.

\emph{Coherence.} The inherited regularity produces a coherence gap of
$0.85\%$ before initiation and $0.70\%$ at initiation, and the gap vanishes
after fracture because the degraded elements carry negligible occupation
weight and the surviving strains no longer span the range on which the
regularity bound binds (Figure~\ref{fig:s3}d). The same constitutive set
yields a coherence effect whose size is dictated by the mechanical state.

\emph{Pure and mixed states.} At $k=7$ a cutting-plane minimisation of the
set energy over pure states, started from the witness state, reached
$0.046691$ against the bracket $[0.046690,0.046691]$, consistent with
equality in \eqref{eq:weak}. The state
space is nonconvex and no strict pure/mixed gap was exhibited. Because the
pure-state search is itself local, this is evidence of absence at the scale
tested, not a proof.

\emph{Information across increments.} At $k=3$ and $k=7$ the discrimination
between the witness maximiser and the softest coherent material seen by the
witness state selected an energy measurement (necessary condition of
Theorem~\ref{thm:gain}, not a certified gain); each contracted the set,
lowered the upper certificate by $0.13\%$ and $0.12\%$, and the bracket re-closed
on the contracted set, which then carried the remaining increments
(Figure~\ref{fig:s3}f). These are realized contractions for the outcome
observed, not guaranteed gains. The mechanics solver and the CSM layer were not
modified.

\emph{Path-level certificate.} On the declared class of irreversibility-admissible
trajectories, the path-level witness loop closed in eight path solves with two
active witnesses; the path oracle is the same alternate-minimisation solver,
so the upper certificate is rigorous and the lower value solver-realized, as
at the increment level. The sequential trajectory lay above the declared-class
trajectory by $7\times10^{-6}$ to $4\times10^{-5}$ at every material visited:
nearly coincident under monotone loading, not identical. This agreement is
reported for this problem and is not generalised.

\emph{What is not claimed.} Certification of arbitrary sequential local trajectories,
of crack-path uniqueness or crack-path bounds, of the global optimality of the
numerical trajectory, or of any path certificate on a $\psi$-dependent
admissibility class. Nor is the stability of the returned states assessed. A separate
second-order stability analysis would require the admissible tangent cone of
the active irreversibility constraints rather than the full Hessian, and it
lies outside the present value-level certificates; branch turning, loss of
admissible second-order stability, and loss of global energetic minimality
are distinct events and are not identified with one another here. CSM
certifies the declared admissible problem.

\subsection{Study 4: an assembled composite shell}\label{sec:s4}

\emph{The question.} Does the method survive contact with a production
finite element code? The previous studies used mechanics written for the
purpose. This one takes an existing, assembled, geometrically nonlinear,
anisotropic composite-shell fracture solver and asks three things: whether
the constitutive-set layer operates on it without modifying the shell
equations, whether the constitutive information the structure interrogates
remains small and state-dependent on a class of several hundred dimensions,
and whether the representation of the material --- analytic, spline or learned
--- can be exchanged without touching the mechanics. It makes no other claim.

\subsubsection{Mechanics, constitutive structure and the socket}

The mechanics is a finite-strain Reissner--Mindlin shell with quadratic
displacements and rotations, linear phase fields, a 9-point surface rule and a
5-point through-thickness rule, plane stress enforced by a local Newton
iteration on $E_{33}$ inside the constitutive evaluation, two surface phase
fields $d_u$, $d_l$ and an interfacial field $\beta$ with obstacle
irreversibility, mode-resolved degradation $g_f$, $g_{if}$ and
$g_fg_{if}g_\beta$ on the fibre, matrix and fibre--matrix channels, and a
single variational functional
$\Pi=\Pi_{\rm mech}+\Pi_{\rm fracture}+\Pi_{\rm interface}-W_{\rm ext}$.
The transversely isotropic energy is decomposed into mechanism potentials on
the orthogonal invariants of the fibre and transverse projectors,
\[
  \Psi=\Psi_f(I_4)+\Psi_{mx}(I_T,J_T)+\Psi_{int}(J_{ft})+\Psi_{\rm closed},
\]
with $I_4$ the fibre stretch invariant, $I_T$ the transverse dilatation,
$J_T=|\mathrm{dev}\,E_T|^2$, $J_{ft}=I_5-I_4^2$, and $\Psi_{\rm closed}$
collecting the isotropic matrix term, the compressive branches and the fracture
energies. The three tensile mechanism potentials are the constitutive input
replaced by a set; everything in $\Psi_{\rm closed}$ enters $c(z)$. The socket
therefore acts on the constitutive operator through its invariant
representation, and it replaces those three potentials and nothing else:
whichever representation is selected, the isotropic matrix term, the
compressive branches and the fracture energies remain closed-form.

Two properties of this mechanics determine the form of the CSM layer. The
first is the case anticipated in Section~\ref{sec:affine}: the density is
affine in the mechanism potentials only on the \emph{enlarged} local state
that includes the condensed $E_{33}$; the condensed density is the
minimum over $E_{33}$ of that affine density and is concave, not affine, in
the constitutive variables. The occupation measure is therefore taken on the
enlarged state, with the condensed $E_{33}$ of each quadrature evaluation
exposed by the constitutive evaluation itself. Second, the potentials must be
at least $C^2$ for the tangent, so the constitutive class is a class of
\emph{cubic B-spline coefficient vectors}: 41 coefficients on $[0,I_4^{\max}]$,
a $12\times12$ tensor-product grid on $[0,I_T^{\max}]\times[0,J_T^{\max}]$, and
41 on $[0,J_{ft}^{\max}]$, $n_{\rm raw}=226$. The energy is affine in the
coefficients, convexity and monotonicity are linear constraints on divided
differences at the Greville abscissae, curvature bounds are exact bounds on
second divided differences, and cubic B-splines reproduce the quadratic and
linear reference potentials exactly, so the reference law is an element of the
class. The three representations --- the analytic law, any coefficient vector
of the class, and a learned law --- enter the mechanics through one socket,
selected by a flag whose default leaves the original evaluation unchanged.

The occupation measure is generated by the production assembly: at every
constitutive evaluation the assembly records the enlarged strain, the
degradation factors, the quadrature weight and the invariants, and
$T(z)$ is the degradation-weighted sum of the B-spline basis vectors at the
recorded invariants. No separate strain sampling is performed.

\subsubsection{Verification of the constitutive socket}

With all extensions at their defaults the assembled energy, residual and
tangent are reproduced to exact floating-point equality. With the spline
socket at the reference coefficients they are reproduced to
$4.7\times10^{-16}$, $4.9\times10^{-12}$ and $1.6\times10^{-13}$ relative
error (energy, residual norm, tangent Frobenius norm) on all eight reference
states. The recorded evaluations reconstruct the assembled mechanical energy
to $2.4\times10^{-16}$, and the pairing $\langle\hat c_{\rm ref},T\rangle+c(z)$
reproduces the assembled energy to $4.7\times10^{-16}$: the occupation measure
is an exact by-product of assembly. Central finite differences of the assembled
energy with respect to twenty displacement degrees of freedom of a state
perturbed by $1\,\mu$m from equilibrium agree with the assembled residual to
$2.0\times10^{-8}$ (median $3.5\times10^{-10}$); at equilibrium the residual
itself is below $10^{-10}$~N, beneath the finite-difference floor.

\subsubsection{Reference response}

A single-edge-notched tension specimen, $12\times24$ structured cells (576
triangles, 7100 degrees of freedom), thickness $1$~mm, seeded notch of
$14.4$~mm, fibre orientation $30^\circ$ to the loading direction; matrix
$E_m=70$~GPa, $\nu_m=0.3$, fibre $E_f=200$~GPa, $\nu_f=0.25$,
$\lambda_{mI}=\mu_{mII}=\mu_{LT}=1$~GPa, $G_c=50$~J/m$^2$, $G_{c,int}=10$~J/m$^2$,
$\ell_c=\ell_\beta=5$~mm, residual stiffness $10^{-7}$. Displacement control to
$20\,\mu$m in eight increments; increments are solved by alternate
minimisation with load-step bisection, the fifth reaching $11.25\,\mu$m of
the scheduled $12.5\,\mu$m. For the reference law the reaction rises to a
peak of $183$~N at $11.25\,\mu$m and the crack traverses the specimen before
$15\,\mu$m, where the reaction has fallen to $28$~N (Figure~\ref{fig:s4a});
this response and the phase fields are those of the reference law and are
illustrative. Bending is not exercised by this configuration. The certified increments are $k=2$ (elastic),
$k=4$ (initiation, the first increment at which a node outside the seeded
notch exceeds $d=0.5$), $k=5$ (peak) and $k=8$ (final).

\begin{figure}[t]\centering
\includegraphics[width=\textwidth]{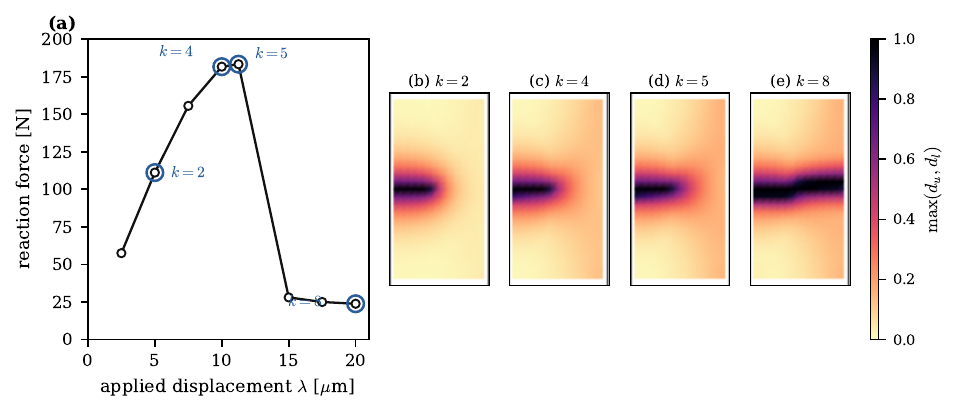}
\caption{Assembled composite shell. (a) Reaction force against applied
displacement for the reference law; the certified increments are ringed.
(b)--(e) Phase field $\max(d_u,d_l)$ on the midsurface at the certified
increments. Both are responses of the reference law and are illustrative;
the certified quantities are the energies of Table~\ref{tab:s4}.}
\label{fig:s4a}
\end{figure}

\subsubsection{Constitutive knowledge}

The set is $\mathcal P=\mathcal P_f\times\mathcal P_{mx}\times\mathcal P_{int}$
with, per mechanism, stress-free normalisation, monotonicity, convexity along
each axis, energy tolerances on every control point alternating between
$3\%$ and $10\%$ about noisy synthetic observations of the reference potential,
and --- for the certified case --- a regularity envelope at the $95\%$ level
inferred from the reference tangent moduli with $4\%$ measurement uncertainty
(curvature bounds on $\Psi_f$ and on $\Psi_{mx}$ in $I_T$, tangent bounds on
$\Psi_{mx}$ in $J_T$ and on $\Psi_{int}$). The reference law is admissible in
both cases (worst constraint violation $4\times10^{-11}$). No cross-mechanism
identity is imposed: the one identity physically available on this material,
between the interaction tangent and the fibre curvature at the origin, does not
bind on the states visited, so the product structure is exact and
$\sigma_{\mathcal P}=\sum_\alpha\sigma_{\mathcal P_\alpha}$ holds to
$1.2\times10^{-7}$ relative at worst, the residual being an optimality gap of
the linear programme rather than a structural defect.

\subsubsection{Certification on the assembled state}

The assembled energy is dominated by the fracture energy of the seeded notch,
$1.887$~J, of which the three mechanism potentials govern between $8$ and
$47\,\mu$J. Certificates are therefore reported on the mechanism-controlled
energy, the only part of $\mathsf W_k$ that depends on the constitutive
input; Table~\ref{tab:s4} and Figure~\ref{fig:s4b}.

\begin{table}[t]\centering\small
\caption{Certification at the selected increments. $\langle\hat c_{\rm ref},T\rangle$
is the mechanism-controlled energy at the reference law, and both bounds are
reported relative to it. The first column of bounds is $U_k-f_k(\psi_{\rm ref})$
after one mechanics solve. At $k=2$ and $k=8$ a second solve was taken and the
value reported is again $U_k-f_k(\psi_{\rm ref})$, the bound that is valid for
any oracle; the bracket at these increments is inverted by the width shown, so
$L_k$ is not a certified lower bound there and is not reported. $n_{\rm vis}$ is
cumulative; $n_{\rm raw}=226$.}
\label{tab:s4}
\begin{tabular}{@{}cccccccc@{}}
\toprule
$k$ & $\lambda$ [$\mu$m] & $\langle\hat c_{\rm ref},T\rangle$ [$\mu$J] & $U_k$, one solve [$\mu$J] & $U_k$, two solves [$\mu$J] & $|U_k-L_k|$ [J] & witnesses & $n_{\rm vis}$ \\
\midrule
2 & 5.00  & 8.45  & 0.282 & 0.2808 & $6\times10^{-11}$ & 1 & 2 \\
4 & 10.00 & 29.07 & 0.932 & --     & --                & 1 & 2 \\
5 & 11.25 & 33.94 & 1.081 & --     & --                & 1 & 3 \\
8 & 20.00 & 46.72 & 1.660 & 1.6392 & $3\times10^{-9}$  & 1 & 4 \\
\bottomrule
\end{tabular}
\end{table}

\begin{figure}[t]\centering
\includegraphics[width=\textwidth]{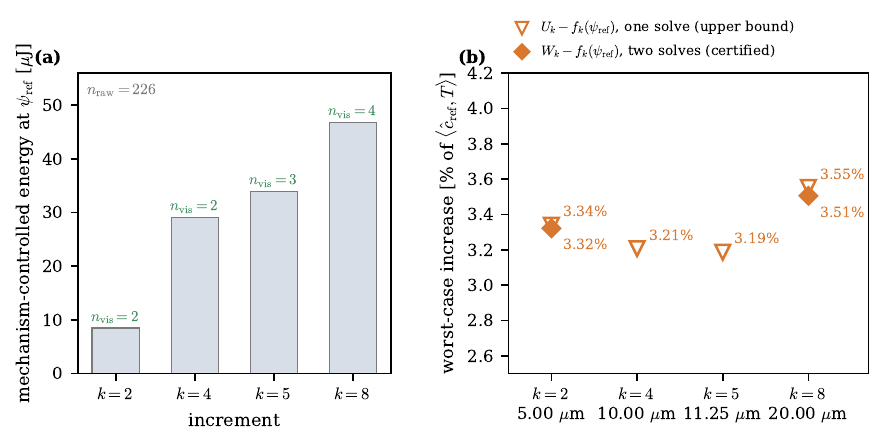}
\caption{Certification on the assembled shell. (a) Mechanism-controlled energy
at the reference law and the cumulative visible dimension against the class
dimension $226$. (b) Worst-case increase of that energy over the admissible
set: the upper bound $U_k$ after one solve and, where taken, after two.}
\label{fig:s4b}
\end{figure}

At every certified increment a single mechanics solve furnishes an upper
bound on the worst-case increase of the mechanism-controlled energy over the
admissible set that lies between $3.2\%$ and $3.6\%$ of that energy. Because
$c(z)$ is constitutively inert and exceeds the mechanism-controlled energy by
five orders of magnitude, the stopping tolerance of the witness loop is scaled
to the mechanism-controlled energy rather than to the total; on the total
energy the tolerance would be met after the first solve at every increment. On
that scaling a second solve at $k=2$ and $k=8$ brings the bracket to
$6\times10^{-11}$~J and $3\times10^{-9}$~J, about $10^{-5}$ of the bounded
quantity, and at both increments the bracket is inverted by that amount: the
local alternate-minimisation oracle returned at the second material a state
marginally better than the bound implied by the first cut, exactly as in
Study~3. Both inversions exceed the optimality gap of the linear programme on
this energy scale (the $1.2\times10^{-7}$ relative gap reported above amounts
to below $10^{-11}$~J) by roughly two to three orders of magnitude, so they
are properties of the oracle and not of the certificate arithmetic. The value
reported after two solves is therefore $U_k$, valid for any oracle by
Theorem~\ref{thm:cert}(i). The dual certificate has one active witness at
every increment; the cumulative visible dimension reaches $4$ of $226$, the
certificate is at all times controlled by the two or three constitutive
directions the shell states interrogate, and this bound concerns the size of
the certificate, not the number of solves.

\subsubsection{Coherence on the assembled state}

At $k=4$ and $k=8$ the coherence gap is evaluated per mechanism
(Figure~\ref{fig:s4c}a). With bounds, convexity and monotonicity only, every gap
is below $10^{-4}\%$: one admissible material is optimal at all $19$--$26$
thousand quadrature atoms simultaneously. With the certified regularity the
fibre potential carries the incoherence, $0.029\%$ at initiation rising to
$0.181\%$ after fracture, the matrix-opening potential contributes
$0.003\%$ after fracture, and the interaction potential none. The phenomenon
of Study~2 is thus evaluated directly on an anisotropic shell state, it is
mechanism-resolved, and it grows with the crack.

\begin{figure}[t]\centering
\includegraphics[width=\textwidth]{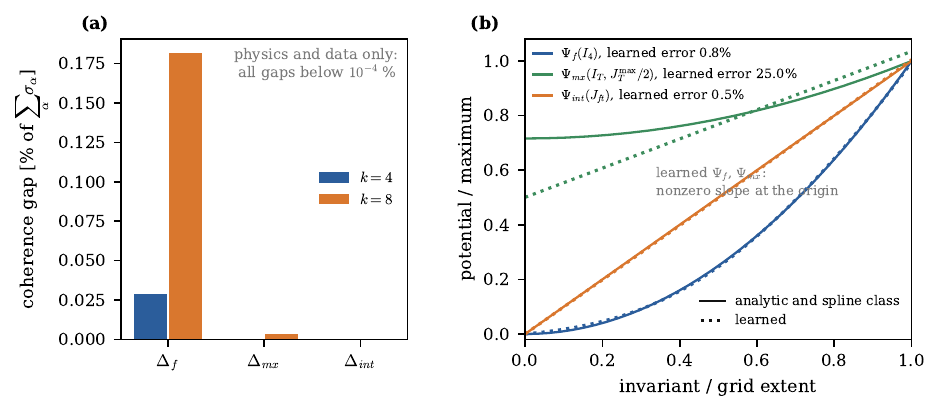}
\caption{Constitutive layer. (a) Mechanism-resolved coherence gap at initiation
and after fracture under the certified regularity; without it all gaps vanish.
(b) The three tensile potentials along their invariants: the analytic law,
which the spline class reproduces to machine precision, and the learned law
with its held-out error.}
\label{fig:s4c}
\end{figure}

\subsubsection{Learned constitutive representation}

One learned representation was placed in the same socket: per mechanism an
input-convex softplus network with eight hidden units, convex and monotone by
construction, fitted to samples of the reference potentials with a held-out
set. The held-out maximum relative errors are $0.79\%$ for $\Psi_f$,
$0.46\%$ for $\Psi_{int}$ and $25\%$ for $\Psi_{mx}$ (the two-dimensional fit;
$5\%$ under a different initialisation). The socket accepted the
representation and evaluated it identically to its training-time definition
($10^{-15}$). The shell, however, did not reach its first increment with it:
the learned $\Psi_f$ and $\Psi_{mx}$ have nonzero derivatives at the origin,
$7.9$~MPa and $10.5$~MPa, so they are not stress-free at zero strain, and the
resulting residual stress perturbs the fracture energy of the seeded notch by
$1.1\times10^{-3}$~J and defeats equilibrium. Projected onto the spline class,
$\Psi_{int}$ lies in $\mathcal P$, whereas $\Psi_f$ and $\Psi_{mx}$ violate its
regularity and tolerance rows by $13\%$ and $23\%$. The result is a
demonstration of the interface and of the admissibility conditions at once:
the representation entered without any change to the shell equations, and the
constitutive knowledge set rejected it for a physical reason ---
stress-free normalisation is a condition on the derivative at the origin,
which a softplus network does not satisfy by construction --- before the
mechanics could be asked to certify anything about it. The message is simple.
The same mechanics can accept an analytic law, a spline representation, or a
learned one; CSM does not decide which representation is used, it decides
which laws are admissible. The network fitted the available data well, and it
violated a basic material requirement, zero stress at zero strain. The
admissible set caught that before the network was used for a structural
prediction, which is the point at which it would otherwise have been caught
by a failed simulation, or not at all.

\subsubsection{Scope}

The study establishes that the assembled nonlinear shell generates the
occupation measure as a by-product of its own assembly; that the witness
certificate operates on it unchanged and is carried by a single equilibrium
state on a constitutive class of dimension $226$; that the interrogated constitutive dimension is
determined by the states the shell visits; that constitutive coherence is
computable and mechanism-resolved on an anisotropic shell state; and that
analytic, spline-coefficient and learned representations enter through one
socket, with admissibility enforced by the set rather than by the mechanics.
It does not establish global optimality of the sequential fracture trajectory,
any certified statement about the crack path, predictive accuracy of the
learned law, or any bending or through-thickness effect.
\subsection{Study 5: constitutive ambiguity revealed by structural loading}\label{sec:s5}

\emph{The question.} A learned constitutive model fits the available data
well. What does the structure know that the fit does not? The preceding
studies establish that the method operates across mechanics of increasing
richness. This one isolates what it adds to a learned constitutive model: a
demonstration that the available evidence, not the representation fitted to
it, determines what a structural prediction is permitted to be, and that the
mechanics determines which part of the remaining ambiguity matters. The
example is chosen to resemble a common testing campaign, in which the axial
and shear responses are characterised well and the coupled multiaxial response
poorly, while the structure loads the material in exactly the coupled region.

\paragraph{Material, data and class.}
The material is a nonlinear anisotropic plane energy $\psi(E_1,E_2,\gamma)$ in
its own axes, with quadratic coupling between the normal strains and quartic
hardening; it plays the role of a hidden truth and is never used by the method.
The constitutive evidence consists of one-mode tests only: seven uniaxial
stress observations along each material axis and seven in shear, each with a
$3\%$ tolerance, and one poor biaxial compression test, four energy
observations with a $25\%$ tolerance in a region of compressive strain. Such
campaigns are common: axial and shear responses are characterised well, the
coupled multiaxial response poorly. The constitutive class is a cubic
B-spline class, $17\times17$ coefficients in $(E_1,E_2)$ on
$[-0.07,0.07]^2$ and $21$ in $|\gamma|$ on $[0,0.08]$, $n_{\rm raw}=310$; the
knowledge set imposes the stress-free origin, convexity by diagonal dominance
of the Hessian and nonnegativity on a $29\times29$ sample grid, the
measurements, and a regularity envelope certified at the $95\%$ level from the
tangent moduli with $4\%$ uncertainty. The reference is admissible.

\paragraph{Learned law.}
An input-convex softplus network with eight hidden units, made stress-free at
the origin by construction (its value and gradient at the origin are subtracted
analytically), is fitted to the same data. Fitted to the data alone it
reproduces the stress observations to $0.7\%$ and violates the certified
regularity fourfold in the region the data do not reach: a good fit that is
not admissible knowledge. Fitted with the regularity and convexity conditions
of the set as penalties it reproduces the observations to $2.8\%$ and lies
within $3\%$ of admissibility; this is the learned law used below.

\paragraph{Structure.}
A plane membrane occupying the unit square with a central hole of radius
$0.22$, $16\times16$ structured triangles with the hole removed, material axes
at $30^\circ$ to the plate axes, symmetry on two edges and prescribed
displacements $\delta$ and $0.6\,\delta$ on the other two, in six increments to
$\delta=0.03$. The states it visits have $E_1\in[-0.008,0.054]$,
$E_2\in[0,0.040]$ and $\gamma\in[-0.046,0.033]$ in the material axes: the
coupled region the one-mode tests leave open, and never the compressive corner
the poor test measured.

\begin{figure}[t]\centering
\includegraphics[width=\textwidth]{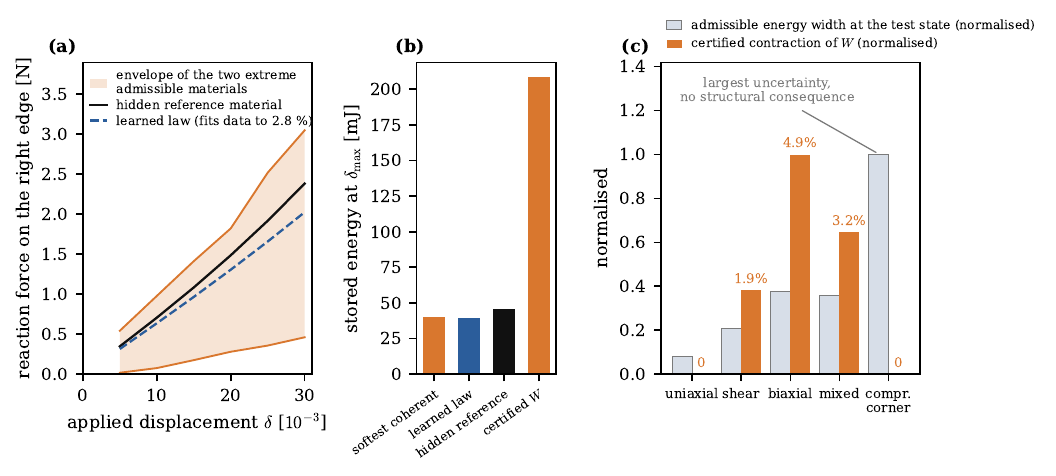}
\caption{Constitutive ambiguity revealed by structural loading. (a) Reaction
force against applied displacement for the hidden reference material, the
learned law, and the two extreme admissible materials whose envelope is
shaded; illustrative. (b) Stored energy at the final increment: softest coherent candidate,
learned law, hidden reference, and the certified upper endpoint. (c) For five
candidate experiments, the width of the admissible energy at the test state
against the realized contraction of the certified upper endpoint when the
experiment is performed on the hidden material; the compression corner carries
the largest constitutive uncertainty and no structural consequence.}
\label{fig:s5}
\end{figure}

\paragraph{Results.}
The learned law predicts a reaction of $2.02$~N at the final increment against
the reference $2.38$~N, a plausible single prediction. Over the admissible
set, the witness loop certifies the upper endpoint of the stored energy at the
final increment in two solves with one active witness, $\mathsf W=0.209$~J
against $0.046$~J for the reference, and extreme-point descent on the concave
value function reaches a softest coherent candidate at $0.040$~J in three
solves. The certified worst-case energy thus exceeds the reference more than
fourfold, and the learned law's stored energy, $0.0396$~J, lies below the
smallest coherent admissible value found: the residual inadmissibility of the
fit places it outside the interval the data support (Figure~\ref{fig:s5}b).
As an illustration of what the two extreme materials mean structurally, their
reaction curves enclose those of the reference and of the learned law and span
$[0.46,3.05]$~N at the final increment (Figure~\ref{fig:s5}a); these curves
are illustrations, not certified bounds on the force. The upper certificate is carried by a
single equilibrium state of a class of dimension $310$.

The result on experiment selection can be stated before the machinery. Five
candidate tests are available. The compression-corner test probes the region
where the admissible laws disagree most, by a wide margin; it is where the
constitutive uncertainty is largest. Yet the structure is almost insensitive
to it, and performing that test changes the certified prediction by
essentially nothing. The biaxial and mixed tests probe a smaller constitutive
uncertainty, in the coupled region the membrane actually enters, and produce
a much larger change in the structurally relevant prediction. The largest
constitutive uncertainty is not the most important uncertainty for this
structure, and it is the mechanics, not the size of the uncertainty, that
tells the two apart.

\begin{table}[t]\centering\footnotesize
\caption{Candidate experiments at the final increment: admissible width of the
energy at the test state, discrimination between the two extreme materials,
and the realized contraction of the certified upper endpoint when the
experiment is performed on the hidden material with a $3\%$ tolerance, as a
fraction of $\mathsf W-f(\psi_{\rm ref})$. The contraction is realized for the
one outcome that material produces, not guaranteed over outcomes.}
\label{tab:s5}
\begin{tabular}{@{}lccc@{}}
\toprule
experiment (material axes) & admissible width [J] & discrimination [J] & realized contraction \\
\midrule
uniaxial, $E_1=0.04$ & $0.013$ & $0.006$ & $0$ \\
pure shear, $\gamma=0.05$ & $0.033$ & $0.033$ & $1.9\%$ \\
biaxial, $E_1=E_2=0.02$ & $0.060$ & $0.056$ & $4.9\%$ \\
mixed, $E_1=0.025$, $E_2=0.012$, $\gamma=0.02$ & $0.058$ & $0.052$ & $3.2\%$ \\
compression corner, $E_1=E_2=-0.05$ & $0.160$ & $0.033$ & $0$ \\
\bottomrule
\end{tabular}
\end{table}

Five candidate experiments are then valued by performing each on the hidden
material and recomputing the certified upper endpoint on the contracted set,
with the cuts of the original certificate carried and the witness loop
re-projected before its first solve (Table~\ref{tab:s5}, Figure~\ref{fig:s5}c).
The uniaxial test, already well represented in the data, contracts nothing.
The compression-corner test has by far the largest admissible width --- it is
where the constitutive uncertainty is largest --- and contracts the certified
endpoint by $2\times10^{-12}$~J: that uncertainty is mechanically invisible to
this structure. The biaxial and mixed tests, which probe the coupled region the
structure actually visits, contract the endpoint most. A criterion that
measures where the constitutive uncertainty is largest selects the corner; the
mechanics selects the coupling. The discrimination between the two extreme
materials, a necessary condition for guaranteed value
(Theorem~\ref{thm:gain}), is nonzero at the corner because the extreme
materials are arbitrary there: it is the realized contraction of the
certified endpoint, not the discrimination, that decides. Each contraction is
realized for the outcome the hidden material produces; the guaranteed
contraction $G_{\Rint}$ over all outcomes would require the lower endpoint on
each intersection, which this study does not certify.

\paragraph{Scope.}
The certified quantity is the upper endpoint of the stored energy. The lower
endpoint is a candidate from extreme-point descent, the reaction curves are
illustrations of the extreme materials and not certified bounds on the force,
and the learned law is one network trained once under stated conditions.

\section{Conclusions}\label{sec:concl}

\emph{What is CSM?} A mechanics framework that carries constitutive
uncertainty as an admissible set of laws instead of selecting one law. The
set is what the evidence permits: thermodynamics, measurements within their
tolerances, certified regularity and known material structure. The mechanics
is unchanged, and no member of the set is ever selected. The replacement requires no modification of the underlying mechanics,
because assembly already makes the incremental energy affine in the
constitutive function through the strain-occupation measure of the state.
That identification is the whole of the method: the energy of a state under
the set is a support function on that measure, the occupation measure of an
equilibrium state is a supergradient of the value function, and the mechanics
solve one performs anyway is the oracle that resolves constitutive
uncertainty.

\emph{What does it add?} It tells us what the structure can certify about
that uncertainty. The same layer was carried without modification from a
truss to a data-constrained nonlinear bar, through a propagating phase-field
crack, onto an assembled finite-strain composite shell, and finally to a
membrane characterised by one-mode tests alone; what changed from study to
study was the mechanics, and the constitutive-set machinery did not.

\emph{What is mechanically visible, and what was learned?} Only the
constitutive directions interrogated by the states the structure generates.
Three things were learned. Certification is mechanical: equilibrium states
are the cuts, the certificate is a mixture of at most $n_{\rm vis}+1$ of
them, and the number of constitutive directions a structure interrogates is
small and set by the states it visits, four of $226$ on the shell and a single
witness on a class of $310$. The upper certificate survives any oracle; the
lower side is a solver-realized value wherever global minimality is not
established, and the inversion of the two proved to be a working detector of
lost globality in a local solver. Coherence is a matter of physics rather than of an assumed budget: the
pointwise-worst material is admissible exactly when the pointwise envelope is
itself an admissible law, closure of the knowledge set under pointwise maximum
being sufficient for that, and it is a certified curvature bound, not a
Lipschitz bound, that removes it. Information is decided by the mechanics: the one-sided guarantee is
vacuous, the guaranteed contraction lives on the response interval, and, for
the outcome observed, the experiment with the largest structural consequence
was not the one where the constitutive uncertainty was largest.

\emph{What does an experiment do?} It contracts the admissible set, and the
mechanics says which contraction matters. The worst-case prediction itself
cannot be guaranteed to fall, because the worst admissible material may be
the true one; the range of predictions can, when the experiment separates the
two extreme admissible materials by more than its tolerance.

\emph{What is the long-term goal?} To use mechanics not only to predict, but
to decide what constitutive information should be acquired next: constitutive
knowledge enters the mechanics, the mechanics identifies the uncertainty that
is relevant to the structure, and that identification directs the next
measurement. The certificates are value-level statements on a declared
incremental problem, and the path-level statements hold for trajectory classes
that do not depend on the material. Two directions follow. Where a constitutive law
selects among competing solution branches, the map from law to branch,
$\psi\mapsto\Gamma(\psi)$, becomes the object to certify; it is neither
concave nor single-valued, and its certification is the open problem for
path-dependent constitutive systems. It is open for a definite reason: once
the material can select the branch or the active set of irreversibility
constraints, the trajectory class ceases to be independent of $\psi$, the
affine pairing of Section~\ref{sec:affine} no longer extends to the path, and
the
stability of the states themselves would have to be assessed on the
admissible tangent cone rather than on the full Hessian. Globally certified sequential
trajectories, and certification of solution fields and crack paths, wait on
that extension. And the quantities engineers act on, a
residual strength or a margin at a design load, are functionals of the
solution rather than values of the energy; extending the certificate from the
value to such outputs, with the same witnesses, is the step that would make
the framework a design tool.

\section*{CRediT authorship contribution statement}
\textbf{Altaf Ahmad Lone:} Writing -- original draft, Visualization,
Validation, Software, Methodology, Investigation, Formal analysis, Data
curation, Conceptualization.
\textbf{R.~T. Durai Prabhakaran:} Writing -- review \& editing, Supervision,
Resources, Project administration.
\textbf{Tawqeer Nasir Tak:} Writing -- review \& editing, Validation,
Methodology.
\textbf{Timon Rabczuk:} Writing -- review \& editing, Supervision,
Methodology, Conceptualization.
\textbf{Xiaoying Zhuang:} Writing -- review \& editing, Supervision,
Methodology.

\section*{Declaration of competing interest}
The authors declare that they have no known competing financial interests or
personal relationships that could have appeared to influence the work reported
in this paper.

\section*{Acknowledgements}
The first author gratefully acknowledges the support of the Indian Institute
of Technology Jammu. The authors thank the Department of Mechanical
Engineering, IIT Jammu, for access to computational resources.

\section*{Data availability}
All quantities reported in Section~\ref{sec:studies} are computed
with fixed seeds by numerical code implementing the interface of
Sections~\ref{sec:seteng} and~\ref{sec:cert} against the respective mechanics
oracle; the linear programmes use primal and dual feasibility tolerances of
$10^{-10}$, and the conic and quadratic subproblems of Study~3 an interior-point
solver. The code and the data supporting the findings of this study are
available from the corresponding author on request.

\end{document}